\documentclass[pra,twocolumn,superscriptaddress,letterpaper,nobalancelastpage]{revtex4-2}

\usepackage{nicefrac}
\usepackage{braket}

\usepackage{upgreek}
\usepackage{graphicx,ulem,bbm}
\usepackage[colorlinks=true,urlcolor=blue,citecolor=blue,linkcolor=blue]{hyperref}

\usepackage{amssymb,amsmath,amsfonts}
\usepackage{bbold}

\usepackage{amsthm}
\newtheorem{lemma}{Lemma}

\renewcommand{\vec}[1] {\mathbf{#1} }
\DeclareMathAlphabet\mathbfcal{OMS}{cmsy}{b}{n}

\newcommand{\dn}[1]{_{\protect\mbox{\protect\scriptsize{\rm #1}}}}
\newcommand{\up}[1]{^{\protect\mbox{\protect\scriptsize{\rm #1}}}}
\newcommand{\mum}{\,\upmu\mbox{\rm m}}
\newcommand{\ie}{\mathrm{i}\,}

\newcommand{\PCe}{\mathrm{e}\,}

\begin{document}

\title{%
    Classical and quantum electromagnetic momentum in anisotropic optical waveguides
}

\author{Denis A. Kopylov}
\email{denis.kopylov@uni-paderborn.de}
\affiliation{Department  of Physics, Paderborn University, 
Paderborn, Germany}
\affiliation{Institute for Photonic Quantum Systems (PhoQS), Paderborn University, 
Paderborn, Germany}
\author{Manfred Hammer}
\affiliation{Theoretical Electrical  Engineering, Paderborn University, 
Paderborn, Germany}
\affiliation{Institute for Photonic Quantum Systems (PhoQS), Paderborn University, 
Paderborn, Germany}

\date{\today}

\begin{abstract}
\noindent
    The guided modes supported by dielectric channel waveguides act as individual carriers of momentum. 
    We show this by proving that the modes satisfy an orthogonality condition which relates to the momentum of the optical electromagnetic field, with a link to the more familiar power (energy) orthogonality. 
    This result forms the basis for a rigorous, self-consistent procedure for the quantization of broadband guided electromagnetic fields in the typical channels used in integrated  photonic circuits.
    Our work removes the existing theoretical gap between the classical solution of the Maxwell equations for guided fields and the intuitive understanding of photons in waveguides.
    The presented approach is valid for straight, lossless, and potentially anisotropic, dielectric waveguides of general shape, in the linear regime, and including material dispersion.
    Examples for the hybrid modes of a thin film lithium niobate strip waveguide are briefly discussed.  

\end{abstract}

\maketitle

\noindent
Dielectric optical waveguides (WGs) \cite{Vassallo_1991_book} constitute the basic building blocks of compact integrated photonic circuits: They play the role of `light-wires', providing the connections between the functional elements on the chip.
Currently, thin film lithium niobate (TFLN) \cite{Poberaj_2012,Boes_2018} is among the promising platforms, in particular also for quantum optical applications \cite{Kai_2019, Nehra_2022, Sund_2023,Arge_2025}. 
Due to the anisotropy of the lithium niobate cores \cite{Weis_1985}, the properties of TFLN channels differ significantly from isotropic WGs, requiring elaborated theoretical models not only for classical \cite{Hammer_2025}, 
but also for quantum fields.

The existence of a set of discrete guided modes---stable independent eigen-solutions of Maxwell equations for the structures with translational invariance---can be seen as the main feature of general straight dielectric WGs. These modes are  
power orthogonal~\cite{Vassallo_1991_book,Jackson_book,Yariv_Yeh_book}, such that energy is preserved and transferred independently by different modes.
In addition to energy, from a fundamental point of view, one can expect that this also holds for the optical momentum. 
So far, however, to the best of our knowledge, a similar orthogonality condition related to the momentum of the guided field appears to be missing.

In the first part of this letter, we introduce a momentum-related orthogonality for the modes of general dielectric WGs. This shows that the guided fields transfer independently not only energy, but also the Minkowski momentum~\cite{Minkowski_1910}.
In the second part, based on the former conditions, we present a self-consistent quantization procedure for guided light in arbitrary transparent dispersive anisotropic channels.
 

\paragraph*{Waveguide modes.}

In line with traditional optical waveguide theory \cite{Vassallo_1991_book}, we consider a waveguide channel for uncharged transparent (lossless) dispersive dielectric and nonmagnetic ($\mu\equiv1$) linear media.
We focus on WGs of arbitrary cross-section shapes with the coordinate system shown in Fig.~\ref{fig_1_scheme}(a), where the full and transverse positions are combined as $\vec{r}\equiv(x, y, z)$ and $\vec{r}_\perp\equiv(x, y)$. 
We explicitly account for an anisotropic medium with a Hermitian tensorial relative permittivity $\tensor{\varepsilon}( \vec{r}_\perp , \omega)$. 
Thus, uniaxial, biaxial, gyro-magnetic as well as isotropic media are covered by our study.
Material dispersion manifests in a frequency dependence of $\tensor{\varepsilon}( \vec{r}_\perp , \omega)$, while the channel shape is encoded in the transverse dependence.   
For an optical electromagnetic (EM) field with a temporal dependence ${\sim}\exp(-\ie \omega t)$, the Maxwell curl equations and material equations \cite{Jackson_book} 
for the electric field $\vec{e}$,
magnetic field $\vec{h}$, 
dielectric displacement $\vec{d}$, 
and magnetic induction $\vec{b}$ are
\begin{equation}
\label{eq_maxwellfreqeq}
        \nabla \times \vec{e}  =  \ie\omega \vec{b}, 
~~ 
        \nabla \times \vec{h}  =  -\ie\omega \vec{d}, 
~~                                                  
        \vec{b} = \mu_0 \vec{h}, 
~~
        \vec{d} = \varepsilon_0 \tensor{\varepsilon}  
\vec{e}.~                                                                        
\end{equation}
Here $\vec{e}$, $\vec{h}$, $\vec{d}$ and $\vec{b}$ depend on $\vec{r}$; the angular frequency $\omega$ plays the role of a parameter;
$\varepsilon_0$ and $\mu_0$ are the vacuum permittivity and permeability, respectively.
In a more compact form, the fields are combined as 
$\vec{x} = \big(\vec{e}, \vec{h}, \vec{d}, \vec{b} \big)$.
\begin{figure}[h]
    \includegraphics[width=1.\columnwidth]{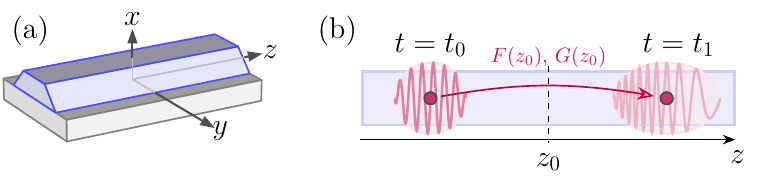}
    \caption{(a) Optical channel WG with translational invariance along its axis $z$. 
    (b) Schematic illustration of an optical pulse in a WG, propagating through the plane at $z_0$.
    It transfers the energy $F(z_0, t_0, t_1)=F(z_0)$ and momentum $G(z_0, t_0, t_1)=G(z_0)$. 
    In the quantum case, the minimal transferred excitation corresponds to a single-photon pulse. 
    }
    \label{fig_1_scheme}
\end{figure}
 
Guided modes are sought as nontrivial solutions of Eqs.~(\ref{eq_maxwellfreqeq}) in modal form
\begin{equation}
    \label{eq_modeeq}
    \vec{x}(\vec{r}) = \mathbfcal{X}(\vec{r}_\perp)~\PCe^{\ie k z},
\end{equation}
with a propagation constant $k$ and an integrable EM mode profile 
$\mathbfcal{X} \equiv  \big(\mathbfcal{E}, \mathbfcal{H}, \mathbfcal{D}, \mathbfcal{B} \big)$ 
that depends on the transverse coordinates $\vec{r}_\perp$ only.
In what follows, we merely require that the modal fields solve Eqs.~(\ref{eq_maxwellfreqeq}). 
Inserted in Eqs.~(\ref{eq_maxwellfreqeq}), the ansatz (\ref{eq_modeeq}) leads to a vectorial eigenvalue problem for the propagation constants $k$, with the mode profiles $\mathbfcal{X}$ 
as eigenfunctions. 
These eigenproblems can be given different forms \cite{Vassallo_1991_book} that do not permit, in general, any analytical treatment.

For the guided modes which are of interest here, there are several conditions:
(i) We assume that a number $\mathcal{M}$ of discrete modes exists, with profiles $\mathbfcal{E}_m$, $\mathbfcal{H}_m$, $\mathbfcal{D}_m$, $\mathbfcal{B}_m$ and related propagation constants $k_m$, for $m \in \bar{\mathcal{M}}$, and $\bar{\mathcal{M}} \equiv \set{1, 2, ..., \mathcal{M}}$.
If a WG does not support any guided modes, then $\bar{\mathcal{M}}= \emptyset$.
(ii) The fields $\mathbfcal{X}_m$ are localized (integrable over $\vec{r}_\perp$).
(iii) The propagation constants $k_m$ are real 
and pairwise different (non-degenerate modes).

The following two orthogonality conditions hold for the guided modes associated with any pair of indices $m, m' \in \bar{\mathcal{M}}$.
Starting from the Lorentz reciprocity theorem \cite{Vassallo_1991_book,Jackson_book,Yariv_Yeh_book}, one can prove the \textit{energy orthogonality condition}
\begin{equation}
     \int d\vec{r}_\perp \: [\mathbfcal{E}^*_m \times \mathbfcal{H}_{m'} + \mathbfcal{E}_{m'} \times \mathbfcal{H}^*_m ]_z = \xi_m\up{E}\,\delta_{m\,{m'} }.
    \label{eq_orthogonality_condition_energy}
\end{equation}
The quantity $\xi\up{E}_m$ (representing, typically, 
four times the optical power associated with the field $m$) is often used for the normalization of the guided modes. 
 
Further, the pairs of modes satisfy a \textit{momentum orthogonality condition}
\begin{multline}
    \int d\vec{r}_\perp ~ \Big[\mathbfcal{E}_m \odot \mathbfcal{D}^*_{m'}  + \mathbfcal{E}^*_{m'} \odot \mathbfcal{D}_m \\ +  \mathbfcal{H}^*_{m'} \odot \mathbfcal{B}_m  + \mathbfcal{H}_m \odot \mathbfcal{B}^*_{m'} \Big] = \xi_m\up{M} \delta_{m\,m'} 
     \label{eq_orthogonality_condition_momentum}
\end{multline}
with the $\odot$-product of two vectors $\vec{X}$, $\vec{Y}$ defined as 
\begin{equation}
\label{eq_odot_definition}
    \vec{X} \odot \vec{Y} \equiv \frac12(X_x Y_x + X_y Y_y - X_z Y_z).
\end{equation}
Here the normalization constant $\xi\up{M}_m$ has the physical dimension of momentum per time.

Moreover, there exists a \textit{relation between the orthogonality conditions} 
\begin{equation}
\dfrac{\xi\up{E}_m} {\xi\up{M}_m}
= \frac{\omega}{k_m}
= \frac{c}{\eta_m}.
   \label{eq_orthogonality_condition_connection}
\end{equation}
Here $\eta_m = k_m c / \omega$ is the effective index associated with mode $m$, for vacuum speed of light $c$. 
The ratio $\xi\up{E}_m/\xi\up{M}_m$ thus represents the phase velocity associated with the 
guided wave \eqref{eq_modeeq}.

Equations~\eqref{eq_orthogonality_condition_momentum} and \eqref{eq_orthogonality_condition_connection} are the main results of this paper, and their proofs are given in Appendix~\ref{app_derivation}.
The following paragraphs discuss their physical interpretation.



\paragraph*{Energy and momentum transfer.}

We first switch to the temporal domain, considering guided   
localized (integrable over $\vec{r}_\perp$) real electric and magnetic fields in space $\vec{r}$ and time $t$, namely, the vector 
$\vec{X} \equiv  \big(\vec{E}, \vec{H}, \vec{D}, \vec{B} \big)$, 
and the EM-field energy 
 \begin{equation}
   \sigma\up{E}(z_0,t) = \int d\vec{r}_\perp S_z(\vec{r}_\perp,z_0,t)
    \label{eq_energy_per_time}
\end{equation}
transferred per unit time through the $xy$-plane at $z_0$.
Here $\int d \vec{r}_\perp$ denotes the $xy$-integral over the infinite cross-sectional plane, and $S_z(\vec{r},t)$ is the energy flow \cite{Jackson_book}, the $z$-projection of the Poynting vector $\vec{S} \equiv \vec{E} \times \vec{H}$.

In addition, we consider the Minkowski momentum~\cite{Minkowski_1910,Pfeifer_2007,Barnett_2010,Griffiths_2011}, whose momentum density is defined as $\vec{p}^{M} \equiv \vec{D} \times \vec{B}$.
The momentum flow is determined by the Minkowski stress-tensor~$\tensor{T}(\vec{r}, t)$.
For the former $xy$-plane, the EM field momentum transferred per unit time is 
\begin{equation}
   \sigma\up{M}(z_0,t) = \int d\vec{r}_\perp \; T_{zz}(\vec{r}_\perp,z_0,t),
    \label{eq_momentum_per_time}
\end{equation}
where the ${zz}$-component of the Minkowski tensor reads 
\begin{equation}
    T_{zz} \equiv \vec{E} \odot \vec{D} + \vec{H} \odot \vec{B}.
\end{equation}
  
\paragraph*{Monochromatic light in a waveguide.}
Let us first consider a monochromatic guided field at frequency $\omega$ propagating in the $z$-direction
\begin{equation}
    \vec{X}(\vec{r}, t) = \sum_{m \in \bar{\mathcal{M}}}  c_m \mathbfcal{X}_m(\vec{r}_\perp, \omega) e^{i k_m(\omega) z} e^{-i\omega  t} + c.c. .
    \label{eq_monochromatic_field_multimode}
\end{equation}
$\vec{X}(\vec{r}, t)$ collects the total fields in the temporal domain, and $\mathbfcal{X}_m(\vec{r}_\perp, \omega)$ is the profile of the $m$-th mode.

For monochromatic waves, it is convenient to consider the time-averaged Poynting vector $\tilde{\vec{S}}(\vec{r})  \equiv \frac{1}{\tau} \int_0^{\tau} dt \; \vec{S}(\vec{r}, t)$, where $\tau = 2\pi/\omega$ is the time period for frequency $\omega$~\cite{Jackson_book}.
Similarly, we consider the time-averaged momentum flow 
$\tilde{T}_{zz}(\vec{r}) \equiv \frac{1}{\tau} \int_0^\tau dt \; T_{zz}(\vec{r}, t)$.
Using Eqs.~\eqref{eq_energy_per_time} and \eqref{eq_momentum_per_time}, the averaged energy $\tilde\sigma\up{E}$
and momentum $\tilde\sigma\up{M}$ transferred by the guided monochromatic light in the form \eqref{eq_monochromatic_field_multimode} trough the $xy$-plane at $z_0$ per unit time are
\begin{equation}
    \tilde\sigma\up{E}(z_0) = \int d\vec{r}_\perp \tilde S_z(\vec{r}) = \sum_{m \in \tilde{\mathcal{M}}}  |c_m|^2 \xi\up{E}_m,
    \label{eq_energy_flow}
\end{equation}
and
\begin{equation}
    \tilde\sigma\up{M}(z_0) = \int d\vec{r}_\perp \; \tilde T_{zz}(\vec{r})  = \sum_{m \in \tilde{\mathcal{M}}} |c_m|^2 \xi\up{M}_m,
    \label{eq_momentum_flow}
\end{equation}
respectively.
Here  
    $\tilde\sigma\up{E}_m \equiv |c_m|^2 \xi\up{E}_m$ and $\tilde\sigma\up{M}_m \equiv |c_m|^2 \xi\up{M}_m$
are the averaged energy and momentum transferred per unit time by the $m$-th mode, respectively.

Obviously, by virtue of the orthogonality properties 
\eqref{eq_orthogonality_condition_energy}
and 
\eqref{eq_orthogonality_condition_momentum}, the separate guided modes contribute individually to the total transfer of energy and momentum. Both the energy and momentum associated with the propagation of each mode as part of the interference field 
\eqref{eq_monochromatic_field_multimode}
are constant along $z$ (independent of the choice of $z_0$).


\paragraph*{Broadband guided pulsed light.}

A general broadband vector field $\vec{X}(\vec{r}, t)$ can be written as a Fourier series and splitted into its positive and negative frequency parts~\cite{Glauber_1963}
\begin{equation}
    \vec{X} = \vec{X}^{(+)} +   \vec{X}^{(-)},
    \label{eq_field_analytical}
\end{equation}
with
\begin{equation}
    \vec{X}^{(+)}(\vec{r}, t) = \frac{1}{\sqrt{2\pi}} \int_0^\infty d\omega e^{-i\omega t} \vec{x}(\vec{r}, \omega) ,
    \label{eq_field_fourier_plus}
\end{equation}
and $ \vec{X}^{(-)}(\vec{r}, t)  = [\vec{X}^{(+)}(\vec{r}, t)]^* $.
For the guided field traveling in the positive $z$-direction, the complex amplitude $\vec{x}$  can be written as
\begin{equation}
    \vec{x}(\vec{r}, \omega) =  \sum_{m \in \bar{\mathcal{M}}(\omega)} c_m(\omega) \mathbfcal{X}_m(\vec{r}_\perp,\omega)e^{ik_m(\omega)z}.
    \label{eq_pulsed_guided_light}
\end{equation}
Note that here all the quantities (mode profiles, propagation constants, the normalization constants $\xi^\cdot_\cdot$, but also the indices $m$) are frequency dependent (modes appear / disappear for varying $\omega$). 
 
For broadband light \eqref{eq_field_fourier_plus}, the time-averaged Poynting vector and Minkowski tensor can not be accurately defined: Averaging over a time period at a central frequency is reasonable only for a narrowband field. 
To find an interpretation of the orthogonality conditions for broadband light, let us first consider the two quantities
\begin{equation}
\label{eq_Fzt}
    F(z_0, t_0, t_1) =  \int_{t_0}^{t_1} dt \int d \vec{r}_\perp \: S_z(\vec{r}_\perp,z_0,t),
\end{equation}
and
\begin{equation}
\label{eq_Gzt}
    G(z_0, t_0, t_1) =  \int_{t_0}^{t_1} dt \int d \vec{r}_\perp \: T_{zz}(\vec{r}_\perp,z_0,t).
\end{equation}
These are the energy and momentum, transferred through the $xy$-plane at $z_0$ within the time interval from $t_0$ to $t_1$, respectively (see Fig.~\ref{fig_1_scheme}(b)).
Next, by taking the limits $t_0 \rightarrow -\infty$ and $t_1 \rightarrow +\infty$, 
and by combining expressions 
\eqref{eq_Fzt} and \eqref{eq_Gzt} with Eqs.~\eqref{eq_field_analytical} and \eqref{eq_pulsed_guided_light}, we obtain
\begin{equation}
    F(z_0) \equiv F(z_0, -\infty, \infty)  =  \int_0^\infty \!\!\! d\omega  \!\!\!  \sum_{m \in \bar{\mathcal{M}}(\omega)}  \!\!\!\!\! |c_m(\omega)|^2 \xi\up{E}_m(\omega),
    \label{eq_transferred_energy}
\end{equation}
and
\begin{equation}
    G(z_0) \equiv G(z_0, -\infty, \infty) = \int_0^\infty \!\!\! d\omega \!\!\! \sum_{m \in \bar{\mathcal{M}}(\omega)} \!\!\!\!\!  |c_m(\omega)|^2 \xi\up{M}_m(\omega).
    \label{eq_momentum_transferred}
\end{equation}
We call $F$ and $G$ the \textit{transferred energy} and \textit{transferred momentum}, respectively.
The values $ \sigma\up{E}_m(\omega) \equiv |c_m(\omega)|^2 \xi\up{E}_m(\omega) $ and $ \sigma\up{M}_m(\omega) \equiv |c_m(\omega)|^2 \xi\up{M}_m(\omega) $ are then the spectral densities of the energy and momentum transferred by the $n$-th mode, respectively.
As a result, for broadband light, the orthogonality conditions \eqref{eq_orthogonality_condition_energy} and \eqref{eq_orthogonality_condition_momentum} guarantee that the transferred energy and the transferred momentum are additive both for different spatial WG modes and different frequencies.


\paragraph*{Field quantization in waveguides.}
 
Typically, the photons in WGs are considered phenomenologically as an abstract bosonic field, i.e., their internal vectorial structure is often not taken into account (see, e.g., Refs.~\cite{Barral_2020,Hamilton_2022,EFH_2022,Sheremet_2023,Delgado_Quesada_2025}).
This comes from the fact that the standard quantization approach~\cite{Mandel_Wolf_book,Vogel_Welsch_book,Raymer_2020_REVIEW,Barnett_2024} meets difficulties when dispersive inhomogeneous media are considered: the corresponding eigenvalue problem for frequencies becomes nonlinear.
Significant simplifications appear when the eigenfunctions are analytically known, e.g., for homogeneous and layered isotropic media~\cite{Huttner_1991,Huttner_1992,Gruner_1996,Bhat_2006,Raabe_2007,Hodgson_2022,Waite_2025} or for homogeneous anisotropic media~\cite{Alekseev_1966,klyshko1988book,Messinger_2020}.
In some applications, the narrowband light approximation is useful (see, e.g., Refs.~\cite{Raymer_2020_REVIEW,Quesada_2020}).
Nevertheless, none of the aforementioned approaches are compatible with the orthogonality conditions in the form of Eqs.~\eqref{eq_orthogonality_condition_energy} and \eqref{eq_orthogonality_condition_momentum}.

Alternatively to the standard approach, the quantization of propagating fields based on temporal modes shows itself to good advantage~\cite{Abram_1987,Huttner_1990,Ben_Aryeh_1991,Perina_1995,Horoshko_2022}.
In particular, for vectorial WG fields, quantization of this type was applied for dispersion-free isotropic WGs~\cite{Li_ares_2003,Li_ares_2008}; however, these assumptions significantly limit possible applications.
Here, we go beyond these strong restrictions.
We show how both the energy and momentum orthogonality conditions lead to a self-consistent quantization procedure for dispersive anisotropic WGs.

From a physical point of view, both the momentum and the energy transferred through some fixed plane at $z_0$ should be quantized.  
Therefore, let us write the guided field in a form that provides a correct discretization of the transferred momentum.
To that end, we write the field operator at $z_0$ as 
\begin{multline}
    \hat{\vec{X}}^{(+)}(\vec{r}_\perp, z_0, t) = \frac{1}{\sqrt{2\pi}} \int_0^\infty d\omega \: e^{-i\omega t} \\ \times \sum_{m \in \bar{\mathcal{M}}(\omega)} \sqrt{\frac{ \hbar k_m(\omega) }{ \xi\up{M}_m(\omega)} } \mathbfcal{X}_m(\vec{r}_\perp, \omega)  \hat{a}_m(z_0,\omega).
    \label{eq_field_operator_fourier_plus_cont}
\end{multline}
Here, the operators $\hat{a}_m(z_0,\omega)$ obey the bosonic commutation relations $[\hat{a}_m(z_0,\omega), \hat{a}_{m^\prime}^\dagger(z_0,\omega^\prime)] = \delta_{m m^\prime}\delta(\omega-\omega^\prime)$.
Using Eq.~\eqref{eq_field_operator_fourier_plus_cont} and $\hat{\vec{X}}^{(-)} = [\hat{\vec{X}}^{(+)}]^\dagger$, the Hermitian transferred momentum operator takes the form
\begin{equation}
    \hat{G} = \int_{0}^{\infty} d\omega  \sum_{m \in \bar{\mathcal{M}}(\omega)}  \hbar k_m(\omega) \left[\hat{a}_m^\dagger(z_0,\omega)\hat{a}_m(z_0,\omega) + \frac12 \right].
    \label{eq_operator_transferred_momentum}
\end{equation}
At the same time, with the use of Eq.~\eqref{eq_orthogonality_condition_connection}, we get the Hermitian transferred energy operator
\begin{equation} 
    \hat{F} = \int_{0}^{\infty} d\omega  \sum_{m \in \bar{\mathcal{M}}(\omega)}  \hbar \omega \left[\hat{a}_m^\dagger(z_0,\omega)\hat{a}_m(z_0,\omega) + \frac12 \right].
    \label{eq_operator_transferred_energy}
\end{equation}
The operator $\hat{n}_m(z_0,\omega) \equiv \hat{a}_m^\dagger(z_0,\omega)\hat{a}_m(z_0,\omega)$ is the photon number density operator in the $m$-th mode; the operator
\begin{equation}
  \hat{N}(z_0) \equiv \int_{0}^{\infty} d\omega  \sum_{m \in \bar{\mathcal{M}}(\omega)}  \hat{n}_m(z_0,\omega)  
\end{equation}
is the total photon number operator of photons crossing the plane $z_0$.
In the same way as in Refs.~\cite{Huttner_1990,Ben_Aryeh_1991}, the Fock-states are build for the fixed plane as the eigenvectors of the photon number density operator $\hat{n}_m({z_0},\omega)$, namely: $\hat{n}_m({z_0},\omega)\ket{n_m({z_0},\omega)} = n_m({z_0},\omega) \ket{n_m({z_0},\omega)} $. 

The diagonal forms of $\hat{F}(z)$ and $\hat{G}(z)$ indicate independent transfer of energy and momentum by monochromatic photons in different WG-modes and different frequencies. 
Moreover, a photon in the mode $m$ with frequency $\omega$ transfers the energy $\hbar \omega$ together with the momentum $\hbar k_m(\omega)$.
Intuitively, this result seems to be trivial; however, the strict proof requires Eqs.~\eqref{eq_orthogonality_condition_energy}, \eqref{eq_orthogonality_condition_momentum} and \eqref{eq_orthogonality_condition_connection}.

Every excitation $\hat{a}_{m}$ corresponds to the simultaneous excitation of four fields $\mathbfcal{X}_{m} = (\mathbfcal{E}, \mathbfcal{H}, \mathbfcal{D}, \mathbfcal{B})$.
Rewriting the fields as $\mathbfcal{B}=\mu_0 \mathbfcal{H}$ and $\mathbfcal{D} = \varepsilon_0 \mathbfcal{E} + \mathbfcal{P}$, we can interpret the field excitation $\hat{a}_m(\omega)$ as a simultaneous excitation of two EM-fields $\mathbfcal{E}$, $\mathbfcal{H}$ together with the polarization vector $\mathbfcal{P}$. 
Therefore, strictly speaking, our field excitations are `waveguide-assisted' \textit{photon-polaritons}.


\paragraph*{ Temporal modes. }

Let us consider broadband photons by introducing broadband frequency modes, known also as temporal modes~\cite{Brecht_2015,Raymer_2020}.
For each spatial mode $m$, we can take the broadband basis $\set{ f^l_m(\omega) }$, with $\int d\omega [f^l_m(\omega)]^*f^{l^\prime}_m(\omega) = \delta_{l {l^\prime}}$. 
The upper indices denote the different broadband frequency modes.
Introducing broadband operators 
\begin{equation}
    \hat{A}^l_m(z_0)  = \int d\omega \ f^l_m(\omega) \hat{a}_m(z_0,\omega), 
    \label{eq_broadband_mode}
\end{equation} 
their commutation relation is $\big[\hat{A}^l_m(z_0),[\hat{A}^{l^\prime}_{m^\prime}(z_0)]^\dagger\big]=\delta_{l {l^\prime}}\delta_{m m^\prime}$.
The $n$-photon Fock state in the broadband mode $f^l_m(\omega)$ is
\begin{equation}
     \ket{n^l_m}_{z_0} = \Big( [\hat{A}^l_m(z_0)]^\dagger \Big)^n \ket{0}_{z_0}.
\end{equation}
The basis vectors for the studied multimode system are constructed as $\bigotimes_{m,n,l } \ket{n^l_m}_{z_0}$.

\begin{figure}
\centering
\includegraphics[width=\linewidth]{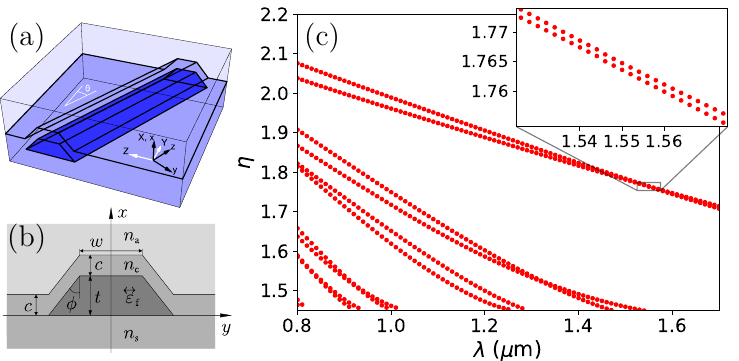}
\caption{Results for an X-cut TFLN strip WG.
    (a) WG schematic; crystal and channel coordinates are 
    $\{X, Y, Z\}$ and   
    $\{x, y, z\}$, respectively.
    (b) WG cross-section; 
        parameters:  $t = 0.6\mum$,  $c = 1\mum$, $w = 0.3445\mum$, sidewall angle $\phi = 30^\circ$, orientation of the channel axis $z$ at angle $\theta = 30^\circ$ versus the crystal $Y$-axis,
        dispersion for SiO$_2$ ($n\dn{s}$ and $n\dn{c}$) and TFLN media (tensor permittivity $\tensor{\varepsilon}\dn{f}$), with air cover ($n\dn{a}$) are as in Refs.~\cite{EFH_2022,EFH_2023,Herzinger_1998}. 
     (c) Effective indices $\eta_m(\omega) = k_m(\omega)c/\omega$ of guided modes as a function of vacuum wavelength $\lambda = 2\pi c/\omega$.
}
\label{fig_2_wg}
\end{figure}

\paragraph*{Spatial Heisenberg equation.}
The transferred momentum operator~\eqref{eq_operator_transferred_momentum} is of great importance since it plays the role of the spatial evolution generator~\cite{Ben_Aryeh_1991,Horoshko_2022,Li_ares_2003,Li_ares_2008}: in other words, it connects fields at different positions $z$. 
In the Heisenberg picture, the operators  obey the spatial Heisenberg equations 
\begin{equation}
    \dfrac{d \hat{a}_m(z,\omega)}{d z} = \frac{1}{i\hbar} \big[ \hat{a}_m(z,\omega), \hat{G} \big] = i k_m(\omega) \hat{a}_m(z,\omega),
\end{equation}
and their solution for  an initial field $\hat{a}_m(z_0,\omega)$ have the form
\begin{equation}
    \hat{a}_m(z,\omega) = \hat{a}_m(z_0,\omega) e^{i k_m(\omega) (z-z_0)}.
\end{equation} 
Note that both WG and material dispersion are embedded into the propagation constant $k_m(\omega)$, leading to dispersion spreading and nonzero chirp (see schematic in Fig.~\ref{fig_1_scheme}(b)).

\paragraph*{Example: An anisotropic TFLN strip WG.}

\begin{figure}
    \centering
    \includegraphics[width=1.\linewidth]{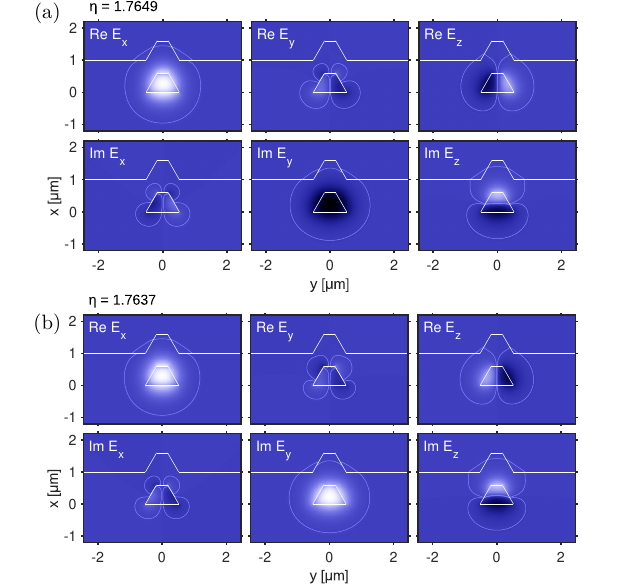}
    \caption{Electric profile components of hybrid (a) first and (b) second order modes at $\lambda = 1.55\mum$, colormap: black-negative, blue-zero, white-positive.  }
    \label{fig_3_profiles}
\end{figure}

Here, we apply our approach to a realistic (adopted from Ref.~\cite{Hammer_2025}) TFLN WG with parameters shown in Fig.~\ref{fig_2_wg}(a,b).
For numerical simulations, we used a commercial finite element solver~\cite{JCMwave}; the material dispersion of LiNbO$_3$ and SiO$_2$ has been taken into account.
As a result, for each frequency $\omega$, we obtained sets of modes $\mathbfcal{X}_{m(\omega)}(\vec{r}_\perp, \omega)$ with corresponding propagation constants $k_{m(\omega)}(\omega)$ (Fig.~\ref{fig_2_wg}(c)).
For the samples at $0.6$, $0.8$, $1.0$, and $1.55\mum$, we calculated the cross-products $\xi\up{E}_{m, m'}$ and $\xi\up{M}_{m, m'}$, defined as the left-hand-sides of Eqs.~\eqref{eq_orthogonality_condition_energy} and
\eqref{eq_orthogonality_condition_momentum}, respectively, for pairs of modes $m$, $m'$ of the same channel.
The numerical accuracy $ 2|\xi\up{E,M}_{m, m'}|/(\xi\up{E,M}_m + \xi\up{E,M}_{m'}) $ was below $10^{-3}$ for all $m \neq m'$. 

According to the quantization procedure, each of the obtained modes $m$ corresponds to the independent monochromatic quantized bosonic field excitation $\hat{a}_m$.
Note that these monochromatic photons obey the non-linear WG-dispersion (see Fig.~\ref{fig_2_wg}(c)), which inherits the material dispersion.

With this example, we want to emphasize that the considered field excitations significantly differ from the well-known photons in vacuum or homogeneous media.
Indeed, the strongly anisotropic core in the considered WG supports vectorial, pronouncedly hybrid modes with essentially complex field profiles (Fig.~\ref{fig_3_profiles}, for details see Ref.~\cite{Hammer_2025}).
These hybrid modes have close propagation constants (Fig.~\ref{fig_2_wg}(c)) and are orthogonal in the sense of the energy and momentum conditions.
Note that the electric profile for both hybrid modes contains not only longitudinal components $\mathcal{E}_z$, but also both pronounced $\mathcal{E}_x$ and $\mathcal{E}_y$ components, indicating that the hybrid modes are not even remotely transverse-electric (TE) nor transverse-magnetic (TM) -like.
Nevertheless, the presented quantization procedure covers these hybrid modes, and the photons inherit the classical field distributions.

\paragraph*{Summary.}

Summing up, we present a momentum orthogonality condition for dielectric WGs.
The connection between this new condition and the well-known energy related orthogonality facilitates a self-consistent quantization scheme for broadband guided EM-fields.
The quantization procedure is based on the solution of the Maxwell equations for guided modes in the frequency domain, and thus it supports both the material- and the WG- dispersion. 
Moreover, it is valid for transparent dielectric WGs with arbitrary cross-sections (including WG arrays and one-dimensional waveguides) with Hermitian dielectric permittivity (including uniaxial, biaxial, gyro-magnetic as well as isotropic media).

What is most important is that the actual field interactions can be included in the quantum description: We can build the quantum spatial propagators for interacting fields (e.g., in the nonlinear optical case) based on an accurate field distribution obtained from the Maxwell equations. 
In particular, with the use of perturbation theory, we can consider the non-ideal lossy WGs~\cite{Hammer_2024} and apply the spatial Langevin equation~\cite{Kopylov_2025}.
Therefore, the introduced quantization scheme provides a fundamental framework for studying WGs for quantum technological tasks.

\begin{acknowledgments}
\noindent
    We thank Torsten Meier, Jens F\"orstner and Polina Sharapova for helpful discussions.
    This work is supported by the Deutsche Forschungsgemeinschaft  (DFG) via the TRR 142/3 (Project No. 231447078, Subproject No. C10).
\end{acknowledgments}

\bibliography{references.bib}

\begin{thebibliography}{50}%
\makeatletter
\providecommand \@ifxundefined [1]{%
 \@ifx{#1\undefined}
}%
\providecommand \@ifnum [1]{%
 \ifnum #1\expandafter \@firstoftwo
 \else \expandafter \@secondoftwo
 \fi
}%
\providecommand \@ifx [1]{%
 \ifx #1\expandafter \@firstoftwo
 \else \expandafter \@secondoftwo
 \fi
}%
\providecommand \natexlab [1]{#1}%
\providecommand \enquote  [1]{``#1''}%
\providecommand \bibnamefont  [1]{#1}%
\providecommand \bibfnamefont [1]{#1}%
\providecommand \citenamefont [1]{#1}%
\providecommand \href@noop [0]{\@secondoftwo}%
\providecommand \href [0]{\begingroup \@sanitize@url \@href}%
\providecommand \@href[1]{\@@startlink{#1}\@@href}%
\providecommand \@@href[1]{\endgroup#1\@@endlink}%
\providecommand \@sanitize@url [0]{\catcode `\\12\catcode `\$12\catcode
  `\&12\catcode `\#12\catcode `\^12\catcode `\_12\catcode `\%12\relax}%
\providecommand \@@startlink[1]{}%
\providecommand \@@endlink[0]{}%
\providecommand \url  [0]{\begingroup\@sanitize@url \@url }%
\providecommand \@url [1]{\endgroup\@href {#1}{\urlprefix }}%
\providecommand \urlprefix  [0]{URL }%
\providecommand \Eprint [0]{\href }%
\providecommand \doibase [0]{https://doi.org/}%
\providecommand \selectlanguage [0]{\@gobble}%
\providecommand \bibinfo  [0]{\@secondoftwo}%
\providecommand \bibfield  [0]{\@secondoftwo}%
\providecommand \translation [1]{[#1]}%
\providecommand \BibitemOpen [0]{}%
\providecommand \bibitemStop [0]{}%
\providecommand \bibitemNoStop [0]{.\EOS\space}%
\providecommand \EOS [0]{\spacefactor3000\relax}%
\providecommand \BibitemShut  [1]{\csname bibitem#1\endcsname}%
\let\auto@bib@innerbib\@empty
\bibitem [{\citenamefont {Vassallo}(1991)}]{Vassallo_1991_book}%
  \BibitemOpen
  \bibfield  {author} {\bibinfo {author} {\bibfnamefont {C.}~\bibnamefont
  {Vassallo}},\ }\href@noop {} {  {\bibinfo {title} {Optical Waveguide
  Concepts}}}\ (\bibinfo  {publisher} {Elsevier},\ \bibinfo {address}
  {Amsterdam},\ \bibinfo {year} {1991})\BibitemShut {NoStop}%
\bibitem [{\citenamefont {Poberaj}\ \emph {et~al.}(2012)\citenamefont
  {Poberaj}, \citenamefont {Hu}, \citenamefont {Sohler},\ and\ \citenamefont
  {G{\"u}nter}}]{Poberaj_2012}%
  \BibitemOpen
  \bibfield  {author} {\bibinfo {author} {\bibfnamefont {G.}~\bibnamefont
  {Poberaj}}, \bibinfo {author} {\bibfnamefont {H.}~\bibnamefont {Hu}},
  \bibinfo {author} {\bibfnamefont {W.}~\bibnamefont {Sohler}},\ and\ \bibinfo
  {author} {\bibfnamefont {P.}~\bibnamefont {G{\"u}nter}},\ }\bibfield  {title}
  {\bibinfo {title} {Lithium niobate on insulator ({LNOI}) for micro-photonic
  devices},\ }\href@noop {} {\bibfield  {journal} {\bibinfo  {journal} {Laser
  \& Photonics Reviews}\ }\textbf {\bibinfo {volume} {6}},\ \bibinfo {pages}
  {488} (\bibinfo {year} {2012})}\BibitemShut {NoStop}%
\bibitem [{\citenamefont {Boes}\ \emph {et~al.}(2018)\citenamefont {Boes},
  \citenamefont {Corcoran}, \citenamefont {Chang}, \citenamefont {Bowers},\
  and\ \citenamefont {Mitchell}}]{Boes_2018}%
  \BibitemOpen
  \bibfield  {author} {\bibinfo {author} {\bibfnamefont {A.}~\bibnamefont
  {Boes}}, \bibinfo {author} {\bibfnamefont {B.}~\bibnamefont {Corcoran}},
  \bibinfo {author} {\bibfnamefont {L.}~\bibnamefont {Chang}}, \bibinfo
  {author} {\bibfnamefont {J.}~\bibnamefont {Bowers}},\ and\ \bibinfo {author}
  {\bibfnamefont {A.}~\bibnamefont {Mitchell}},\ }\bibfield  {title} {\bibinfo
  {title} {Status and potential of lithium niobate on insulator ({LNOI}) for
  photonic integrated circuits},\ }\href@noop {} {\bibfield  {journal}
  {\bibinfo  {journal} {Laser \& Photonics Reviews}\ }\textbf {\bibinfo
  {volume} {12}},\ \bibinfo {pages} {1700256} (\bibinfo {year}
  {2018})}\BibitemShut {NoStop}%
\bibitem [{\citenamefont {Luo}\ \emph {et~al.}(2019)\citenamefont {Luo},
  \citenamefont {Brauner}, \citenamefont {Eigner}, \citenamefont {Sharapova},
  \citenamefont {Ricken}, \citenamefont {Meier}, \citenamefont {Herrmann},\
  and\ \citenamefont {Silberhorn}}]{Kai_2019}%
  \BibitemOpen
  \bibfield  {author} {\bibinfo {author} {\bibfnamefont {K.-H.}\ \bibnamefont
  {Luo}}, \bibinfo {author} {\bibfnamefont {S.}~\bibnamefont {Brauner}},
  \bibinfo {author} {\bibfnamefont {C.}~\bibnamefont {Eigner}}, \bibinfo
  {author} {\bibfnamefont {P.~R.}\ \bibnamefont {Sharapova}}, \bibinfo {author}
  {\bibfnamefont {R.}~\bibnamefont {Ricken}}, \bibinfo {author} {\bibfnamefont
  {T.}~\bibnamefont {Meier}}, \bibinfo {author} {\bibfnamefont
  {H.}~\bibnamefont {Herrmann}},\ and\ \bibinfo {author} {\bibfnamefont
  {C.}~\bibnamefont {Silberhorn}},\ }\bibfield  {title} {\bibinfo {title}
  {Nonlinear integrated quantum electro-optic circuits},\ }\href
  {https://doi.org/10.1126/sciadv.aat1451} {\bibfield  {journal} {\bibinfo
  {journal} {Science Advances}\ }\textbf {\bibinfo {volume} {5}},\ \bibinfo
  {pages} {eaat1451} (\bibinfo {year} {2019})}\BibitemShut {NoStop}%
\bibitem [{\citenamefont {Nehra}\ \emph {et~al.}(2022)\citenamefont {Nehra},
  \citenamefont {Sekine}, \citenamefont {Ledezma}, \citenamefont {Guo},
  \citenamefont {Gray}, \citenamefont {Roy},\ and\ \citenamefont
  {Marandi}}]{Nehra_2022}%
  \BibitemOpen
  \bibfield  {author} {\bibinfo {author} {\bibfnamefont {R.}~\bibnamefont
  {Nehra}}, \bibinfo {author} {\bibfnamefont {R.}~\bibnamefont {Sekine}},
  \bibinfo {author} {\bibfnamefont {L.}~\bibnamefont {Ledezma}}, \bibinfo
  {author} {\bibfnamefont {Q.}~\bibnamefont {Guo}}, \bibinfo {author}
  {\bibfnamefont {R.~M.}\ \bibnamefont {Gray}}, \bibinfo {author}
  {\bibfnamefont {A.}~\bibnamefont {Roy}},\ and\ \bibinfo {author}
  {\bibfnamefont {A.}~\bibnamefont {Marandi}},\ }\bibfield  {title} {\bibinfo
  {title} {Few-cycle vacuum squeezing in nanophotonics},\ }\href
  {https://doi.org/10.1126/science.abo6213} {\bibfield  {journal} {\bibinfo
  {journal} {Science}\ }\textbf {\bibinfo {volume} {377}},\ \bibinfo {pages}
  {1333} (\bibinfo {year} {2022})}\BibitemShut {NoStop}%
\bibitem [{\citenamefont {Sund}\ \emph {et~al.}(2023)\citenamefont {Sund},
  \citenamefont {Lomonte}, \citenamefont {Paesani}, \citenamefont {Wang},
  \citenamefont {Carolan}, \citenamefont {Bart}, \citenamefont {Wieck},
  \citenamefont {Ludwig}, \citenamefont {Midolo}, \citenamefont {Pernice},
  \citenamefont {Lodahl},\ and\ \citenamefont {Lenzini}}]{Sund_2023}%
  \BibitemOpen
  \bibfield  {author} {\bibinfo {author} {\bibfnamefont {P.~I.}\ \bibnamefont
  {Sund}}, \bibinfo {author} {\bibfnamefont {E.}~\bibnamefont {Lomonte}},
  \bibinfo {author} {\bibfnamefont {S.}~\bibnamefont {Paesani}}, \bibinfo
  {author} {\bibfnamefont {Y.}~\bibnamefont {Wang}}, \bibinfo {author}
  {\bibfnamefont {J.}~\bibnamefont {Carolan}}, \bibinfo {author} {\bibfnamefont
  {N.}~\bibnamefont {Bart}}, \bibinfo {author} {\bibfnamefont {A.~D.}\
  \bibnamefont {Wieck}}, \bibinfo {author} {\bibfnamefont {A.}~\bibnamefont
  {Ludwig}}, \bibinfo {author} {\bibfnamefont {L.}~\bibnamefont {Midolo}},
  \bibinfo {author} {\bibfnamefont {W.~H.}\ \bibnamefont {Pernice}}, \bibinfo
  {author} {\bibfnamefont {P.}~\bibnamefont {Lodahl}},\ and\ \bibinfo {author}
  {\bibfnamefont {F.}~\bibnamefont {Lenzini}},\ }\bibfield  {title} {\bibinfo
  {title} {High-speed thin-film lithium niobate quantum processor driven by a
  solid-state quantum emitter},\ }\href
  {https://doi.org/10.1126/sciadv.adg7268} {\bibfield  {journal} {\bibinfo
  {journal} {Science Advances}\ }\textbf {\bibinfo {volume} {9}},\ \bibinfo
  {pages} {eadg7268} (\bibinfo {year} {2023})}\BibitemShut {NoStop}%
\bibitem [{\citenamefont {Arge}\ \emph {et~al.}(2025)\citenamefont {Arge},
  \citenamefont {Jo}, \citenamefont {Nguyen}, \citenamefont {Lenzini},
  \citenamefont {Lomonte}, \citenamefont {Nielsen}, \citenamefont
  {Domeneguetti}, \citenamefont {Neergaard-Nielsen}, \citenamefont {Pernice},
  \citenamefont {Gehring},\ and\ \citenamefont {Andersen}}]{Arge_2025}%
  \BibitemOpen
  \bibfield  {author} {\bibinfo {author} {\bibfnamefont {T.~N.}\ \bibnamefont
  {Arge}}, \bibinfo {author} {\bibfnamefont {S.}~\bibnamefont {Jo}}, \bibinfo
  {author} {\bibfnamefont {H.~Q.}\ \bibnamefont {Nguyen}}, \bibinfo {author}
  {\bibfnamefont {F.}~\bibnamefont {Lenzini}}, \bibinfo {author} {\bibfnamefont
  {E.}~\bibnamefont {Lomonte}}, \bibinfo {author} {\bibfnamefont {J.~A.~H.}\
  \bibnamefont {Nielsen}}, \bibinfo {author} {\bibfnamefont {R.~R.}\
  \bibnamefont {Domeneguetti}}, \bibinfo {author} {\bibfnamefont {J.~S.}\
  \bibnamefont {Neergaard-Nielsen}}, \bibinfo {author} {\bibfnamefont
  {W.}~\bibnamefont {Pernice}}, \bibinfo {author} {\bibfnamefont
  {T.}~\bibnamefont {Gehring}},\ and\ \bibinfo {author} {\bibfnamefont {U.~L.}\
  \bibnamefont {Andersen}},\ }\bibfield  {title} {\bibinfo {title}
  {Demonstration of a squeezed light source on thin-film lithium niobate with
  modal phase matching},\ }\href {https://doi.org/10.1364/OPTICAQ.562545}
  {\bibfield  {journal} {\bibinfo  {journal} {Optica Quantum}\ }\textbf
  {\bibinfo {volume} {3}},\ \bibinfo {pages} {467} (\bibinfo {year}
  {2025})}\BibitemShut {NoStop}%
\bibitem [{\citenamefont {Weis}\ and\ \citenamefont
  {Gaylord}(1985)}]{Weis_1985}%
  \BibitemOpen
  \bibfield  {author} {\bibinfo {author} {\bibfnamefont {R.~S.}\ \bibnamefont
  {Weis}}\ and\ \bibinfo {author} {\bibfnamefont {T.~K.}\ \bibnamefont
  {Gaylord}},\ }\bibfield  {title} {\bibinfo {title} {Lithium niobate: Summary
  of physical properties and crystal structure},\ }\href@noop {} {\bibfield
  {journal} {\bibinfo  {journal} {Applied Physics A}\ }\textbf {\bibinfo
  {volume} {37}},\ \bibinfo {pages} {191} (\bibinfo {year} {1985})}\BibitemShut
  {NoStop}%
\bibitem [{\citenamefont {Hammer}\ \emph {et~al.}(2025)\citenamefont {Hammer},
  \citenamefont {Khan}, \citenamefont {Taheri}, \citenamefont {Farheen},\ and\
  \citenamefont {Förstner}}]{Hammer_2025}%
  \BibitemOpen
  \bibfield  {author} {\bibinfo {author} {\bibfnamefont {M.}~\bibnamefont
  {Hammer}}, \bibinfo {author} {\bibfnamefont {S.~S.}\ \bibnamefont {Khan}},
  \bibinfo {author} {\bibfnamefont {B.}~\bibnamefont {Taheri}}, \bibinfo
  {author} {\bibfnamefont {H.}~\bibnamefont {Farheen}},\ and\ \bibinfo {author}
  {\bibfnamefont {J.}~\bibnamefont {Förstner}},\ }\bibfield  {title} {\bibinfo
  {title} {{TFLN} channel waveguides of rib and strip type: properties of
  guided modes},\ }\href {https://doi.org/10.1364/optcon.569959} {\bibfield
  {journal} {\bibinfo  {journal} {Optics Continuum}\ }\textbf {\bibinfo
  {volume} {4}},\ \bibinfo {pages} {2356} (\bibinfo {year} {2025})}\BibitemShut
  {NoStop}%
\bibitem [{\citenamefont {Jackson}(1998)}]{Jackson_book}%
  \BibitemOpen
  \bibfield  {author} {\bibinfo {author} {\bibfnamefont {J.~D.}\ \bibnamefont
  {Jackson}},\ }\href@noop {} {  {\bibinfo {title} {Classical
  Electrodynamics}}},\ \bibinfo {edition} {3rd}\ ed.\ (\bibinfo  {publisher}
  {John Wiley \& Sons},\ \bibinfo {address} {Nashville, TN},\ \bibinfo {year}
  {1998})\BibitemShut {NoStop}%
\bibitem [{\citenamefont {Yariv}\ and\ \citenamefont
  {Yeh}(1984)}]{Yariv_Yeh_book}%
  \BibitemOpen
  \bibfield  {author} {\bibinfo {author} {\bibfnamefont {A.}~\bibnamefont
  {Yariv}}\ and\ \bibinfo {author} {\bibfnamefont {P.}~\bibnamefont {Yeh}},\
  }\href@noop {} {  {\bibinfo {title} {Optical Waves in Crystals:
  Propagation and Control of Laser Radiation}}}\ (\bibinfo  {publisher}
  {Wiley},\ \bibinfo {address} {New York},\ \bibinfo {year} {1984})\BibitemShut
  {NoStop}%
\bibitem [{\citenamefont {Minkowski}(1910)}]{Minkowski_1910}%
  \BibitemOpen
  \bibfield  {author} {\bibinfo {author} {\bibfnamefont {H.}~\bibnamefont
  {Minkowski}},\ }\bibfield  {title} {\bibinfo {title} {Die {G}rundgleichungen
  fuer die elektromagnetischen {V}orgaenge in bewegten {K}oerpern},\ }\href
  {https://doi.org/10.1007/bf01455871} {\bibfield  {journal} {\bibinfo
  {journal} {Mathematische Annalen}\ }\textbf {\bibinfo {volume} {68}},\
  \bibinfo {pages} {472–525} (\bibinfo {year} {1910})}\BibitemShut {NoStop}%
\bibitem [{\citenamefont {Pfeifer}\ \emph {et~al.}(2007)\citenamefont
  {Pfeifer}, \citenamefont {Nieminen}, \citenamefont {Heckenberg},\ and\
  \citenamefont {Rubinsztein-Dunlop}}]{Pfeifer_2007}%
  \BibitemOpen
  \bibfield  {author} {\bibinfo {author} {\bibfnamefont {R.~N.~C.}\
  \bibnamefont {Pfeifer}}, \bibinfo {author} {\bibfnamefont {T.~A.}\
  \bibnamefont {Nieminen}}, \bibinfo {author} {\bibfnamefont {N.~R.}\
  \bibnamefont {Heckenberg}},\ and\ \bibinfo {author} {\bibfnamefont
  {H.}~\bibnamefont {Rubinsztein-Dunlop}},\ }\bibfield  {title} {\bibinfo
  {title} {Colloquium: Momentum of an electromagnetic wave in dielectric
  media},\ }\href {https://doi.org/10.1103/revmodphys.79.1197} {\bibfield
  {journal} {\bibinfo  {journal} {Reviews of Modern Physics}\ }\textbf
  {\bibinfo {volume} {79}},\ \bibinfo {pages} {1197–1216} (\bibinfo {year}
  {2007})}\BibitemShut {NoStop}%
\bibitem [{\citenamefont {Barnett}(2010)}]{Barnett_2010}%
  \BibitemOpen
  \bibfield  {author} {\bibinfo {author} {\bibfnamefont {S.~M.}\ \bibnamefont
  {Barnett}},\ }\bibfield  {title} {\bibinfo {title} {Resolution of the
  {A}braham-{M}inkowski dilemma},\ }\href
  {https://doi.org/10.1103/physrevlett.104.070401} {\bibfield  {journal}
  {\bibinfo  {journal} {Physical Review Letters}\ }\textbf {\bibinfo {volume}
  {104}},\ \bibinfo {pages} {070401} (\bibinfo {year} {2010})}\BibitemShut
  {NoStop}%
\bibitem [{\citenamefont {Griffiths}(2011)}]{Griffiths_2011}%
  \BibitemOpen
  \bibfield  {author} {\bibinfo {author} {\bibfnamefont {D.~J.}\ \bibnamefont
  {Griffiths}},\ }\bibfield  {title} {\bibinfo {title} {Resource letter em-1:
  Electromagnetic momentum},\ }\href {https://doi.org/10.1119/1.3641979}
  {\bibfield  {journal} {\bibinfo  {journal} {American Journal of Physics}\
  }\textbf {\bibinfo {volume} {80}},\ \bibinfo {pages} {7–18} (\bibinfo
  {year} {2011})}\BibitemShut {NoStop}%
\bibitem [{\citenamefont {Glauber}(1963)}]{Glauber_1963}%
  \BibitemOpen
  \bibfield  {author} {\bibinfo {author} {\bibfnamefont {R.~J.}\ \bibnamefont
  {Glauber}},\ }\bibfield  {title} {\bibinfo {title} {The quantum theory of
  optical coherence},\ }\href {https://doi.org/10.1103/physrev.130.2529}
  {\bibfield  {journal} {\bibinfo  {journal} {Physical Review}\ }\textbf
  {\bibinfo {volume} {130}},\ \bibinfo {pages} {2529} (\bibinfo {year}
  {1963})}\BibitemShut {NoStop}%
\bibitem [{\citenamefont {Barral}\ \emph {et~al.}(2020)\citenamefont {Barral},
  \citenamefont {Walschaers}, \citenamefont {Bencheikh}, \citenamefont
  {Parigi}, \citenamefont {Levenson}, \citenamefont {Treps},\ and\
  \citenamefont {Belabas}}]{Barral_2020}%
  \BibitemOpen
  \bibfield  {author} {\bibinfo {author} {\bibfnamefont {D.}~\bibnamefont
  {Barral}}, \bibinfo {author} {\bibfnamefont {M.}~\bibnamefont {Walschaers}},
  \bibinfo {author} {\bibfnamefont {K.}~\bibnamefont {Bencheikh}}, \bibinfo
  {author} {\bibfnamefont {V.}~\bibnamefont {Parigi}}, \bibinfo {author}
  {\bibfnamefont {J.~A.}\ \bibnamefont {Levenson}}, \bibinfo {author}
  {\bibfnamefont {N.}~\bibnamefont {Treps}},\ and\ \bibinfo {author}
  {\bibfnamefont {N.}~\bibnamefont {Belabas}},\ }\bibfield  {title} {\bibinfo
  {title} {Quantum state engineering in arrays of nonlinear waveguides},\
  }\href {https://doi.org/10.1103/physreva.102.043706} {\bibfield  {journal}
  {\bibinfo  {journal} {Physical Review A}\ }\textbf {\bibinfo {volume}
  {102}},\ \bibinfo {pages} {043706} (\bibinfo {year} {2020})}\BibitemShut
  {NoStop}%
\bibitem [{\citenamefont {Hamilton}\ \emph {et~al.}(2022)\citenamefont
  {Hamilton}, \citenamefont {Christ}, \citenamefont {Barkhofen}, \citenamefont
  {Barnett}, \citenamefont {Jex},\ and\ \citenamefont
  {Silberhorn}}]{Hamilton_2022}%
  \BibitemOpen
  \bibfield  {author} {\bibinfo {author} {\bibfnamefont {C.~S.}\ \bibnamefont
  {Hamilton}}, \bibinfo {author} {\bibfnamefont {R.}~\bibnamefont {Christ}},
  \bibinfo {author} {\bibfnamefont {S.}~\bibnamefont {Barkhofen}}, \bibinfo
  {author} {\bibfnamefont {S.~M.}\ \bibnamefont {Barnett}}, \bibinfo {author}
  {\bibfnamefont {I.}~\bibnamefont {Jex}},\ and\ \bibinfo {author}
  {\bibfnamefont {C.}~\bibnamefont {Silberhorn}},\ }\bibfield  {title}
  {\bibinfo {title} {Quantum-state creation in nonlinear-waveguide arrays},\
  }\href {https://doi.org/10.1103/physreva.105.042622} {\bibfield  {journal}
  {\bibinfo  {journal} {Physical Review A}\ }\textbf {\bibinfo {volume}
  {105}},\ \bibinfo {pages} {042622} (\bibinfo {year} {2022})}\BibitemShut
  {NoStop}%
\bibitem [{\citenamefont {Ebers}\ \emph {et~al.}(2022)\citenamefont {Ebers},
  \citenamefont {Ferreri}, \citenamefont {Hammer}, \citenamefont {Albert},
  \citenamefont {Meier}, \citenamefont {F\"{o}rstner},\ and\ \citenamefont
  {Sharapova}}]{EFH_2022}%
  \BibitemOpen
  \bibfield  {author} {\bibinfo {author} {\bibfnamefont {L.}~\bibnamefont
  {Ebers}}, \bibinfo {author} {\bibfnamefont {A.}~\bibnamefont {Ferreri}},
  \bibinfo {author} {\bibfnamefont {M.}~\bibnamefont {Hammer}}, \bibinfo
  {author} {\bibfnamefont {M.}~\bibnamefont {Albert}}, \bibinfo {author}
  {\bibfnamefont {C.}~\bibnamefont {Meier}}, \bibinfo {author} {\bibfnamefont
  {J.}~\bibnamefont {F\"{o}rstner}},\ and\ \bibinfo {author} {\bibfnamefont
  {P.~R.}\ \bibnamefont {Sharapova}},\ }\bibfield  {title} {\bibinfo {title}
  {Flexible source of correlated photons based on {LNOI} rib waveguides},\
  }\href@noop {} {\bibfield  {journal} {\bibinfo  {journal} {Journal of
  Physics: Photonics}\ }\textbf {\bibinfo {volume} {4}},\ \bibinfo {pages}
  {025001} (\bibinfo {year} {2022})}\BibitemShut {NoStop}%
\bibitem [{\citenamefont {Sheremet}\ \emph {et~al.}(2023)\citenamefont
  {Sheremet}, \citenamefont {Petrov}, \citenamefont {Iorsh}, \citenamefont
  {Poshakinskiy},\ and\ \citenamefont {Poddubny}}]{Sheremet_2023}%
  \BibitemOpen
  \bibfield  {author} {\bibinfo {author} {\bibfnamefont {A.~S.}\ \bibnamefont
  {Sheremet}}, \bibinfo {author} {\bibfnamefont {M.~I.}\ \bibnamefont
  {Petrov}}, \bibinfo {author} {\bibfnamefont {I.~V.}\ \bibnamefont {Iorsh}},
  \bibinfo {author} {\bibfnamefont {A.~V.}\ \bibnamefont {Poshakinskiy}},\ and\
  \bibinfo {author} {\bibfnamefont {A.~N.}\ \bibnamefont {Poddubny}},\
  }\bibfield  {title} {\bibinfo {title} {Waveguide quantum electrodynamics:
  Collective radiance and photon-photon correlations},\ }\href
  {https://doi.org/10.1103/revmodphys.95.015002} {\bibfield  {journal}
  {\bibinfo  {journal} {Reviews of Modern Physics}\ }\textbf {\bibinfo {volume}
  {95}},\ \bibinfo {pages} {015002} (\bibinfo {year} {2023})}\BibitemShut
  {NoStop}%
\bibitem [{\citenamefont {Delgado-Quesada}\ \emph {et~al.}(2025)\citenamefont
  {Delgado-Quesada}, \citenamefont {Barral}, \citenamefont {Bencheikh},\ and\
  \citenamefont {Rojas-González}}]{Delgado_Quesada_2025}%
  \BibitemOpen
  \bibfield  {author} {\bibinfo {author} {\bibfnamefont {J.}~\bibnamefont
  {Delgado-Quesada}}, \bibinfo {author} {\bibfnamefont {D.}~\bibnamefont
  {Barral}}, \bibinfo {author} {\bibfnamefont {K.}~\bibnamefont {Bencheikh}},\
  and\ \bibinfo {author} {\bibfnamefont {E.~A.}\ \bibnamefont
  {Rojas-González}},\ }\bibfield  {title} {\bibinfo {title} {Analytic solution
  to degenerate biphoton states generated in arrays of nonlinear waveguides},\
  }\href {https://doi.org/10.1103/n9yc-kz52} {\bibfield  {journal} {\bibinfo
  {journal} {Physical Review A}\ }\textbf {\bibinfo {volume} {112}},\ \bibinfo
  {pages} {013703} (\bibinfo {year} {2025})}\BibitemShut {NoStop}%
\bibitem [{\citenamefont {Mandel}\ and\ \citenamefont
  {Wolf}(1995)}]{Mandel_Wolf_book}%
  \BibitemOpen
  \bibfield  {author} {\bibinfo {author} {\bibfnamefont {L.}~\bibnamefont
  {Mandel}}\ and\ \bibinfo {author} {\bibfnamefont {E.}~\bibnamefont {Wolf}},\
  }\href {https://doi.org/10.1017/CBO9781139644105} {\emph {\bibinfo {title}
  {{Optical Coherence and Quantum Optics}}}}\ (\bibinfo  {publisher} {Cambridge
  University Press},\ \bibinfo {year} {1995})\BibitemShut {NoStop}%
\bibitem [{\citenamefont {Vogel}\ and\ \citenamefont
  {Welsch}(2006)}]{Vogel_Welsch_book}%
  \BibitemOpen
  \bibfield  {author} {\bibinfo {author} {\bibfnamefont {W.}~\bibnamefont
  {Vogel}}\ and\ \bibinfo {author} {\bibfnamefont {D.-G.}\ \bibnamefont
  {Welsch}},\ }\href@noop {} {  {\bibinfo {title} {Quantum Optics}}}\
  (\bibinfo  {publisher} {Wiley-VCH, Berlin},\ \bibinfo {year}
  {2006})\BibitemShut {NoStop}%
\bibitem [{\citenamefont {Raymer}(2020)}]{Raymer_2020_REVIEW}%
  \BibitemOpen
  \bibfield  {author} {\bibinfo {author} {\bibfnamefont {M.~G.}\ \bibnamefont
  {Raymer}},\ }\bibfield  {title} {\bibinfo {title} {Quantum theory of light in
  a dispersive structured linear dielectric: a macroscopic {H}amiltonian
  tutorial treatment},\ }\href {https://doi.org/10.1080/09500340.2019.1706773}
  {\bibfield  {journal} {\bibinfo  {journal} {Journal of Modern Optics}\
  }\textbf {\bibinfo {volume} {67}},\ \bibinfo {pages} {196–212} (\bibinfo
  {year} {2020})}\BibitemShut {NoStop}%
\bibitem [{\citenamefont {Barnett}(2024)}]{Barnett_2024}%
  \BibitemOpen
  \bibfield  {author} {\bibinfo {author} {\bibfnamefont {S.~M.}\ \bibnamefont
  {Barnett}},\ }\bibfield  {title} {\bibinfo {title} {The quantum optics of
  media},\ }\href {https://doi.org/10.1098/rsta.2023.0339} {\bibfield
  {journal} {\bibinfo  {journal} {Philosophical Transactions of the Royal
  Society A: Mathematical, Physical and Engineering Sciences}\ }\textbf
  {\bibinfo {volume} {382}},\ \bibinfo {pages} {20230339} (\bibinfo {year}
  {2024})}\BibitemShut {NoStop}%
\bibitem [{\citenamefont {Huttner}\ \emph {et~al.}(1991)\citenamefont
  {Huttner}, \citenamefont {Baumberg},\ and\ \citenamefont
  {Barnett}}]{Huttner_1991}%
  \BibitemOpen
  \bibfield  {author} {\bibinfo {author} {\bibfnamefont {B.}~\bibnamefont
  {Huttner}}, \bibinfo {author} {\bibfnamefont {J.~J.}\ \bibnamefont
  {Baumberg}},\ and\ \bibinfo {author} {\bibfnamefont {S.~M.}\ \bibnamefont
  {Barnett}},\ }\bibfield  {title} {\bibinfo {title} {Canonical quantization of
  light in a linear dielectric},\ }\href
  {https://doi.org/10.1209/0295-5075/16/2/010} {\bibfield  {journal} {\bibinfo
  {journal} {Europhysics Letters (EPL)}\ }\textbf {\bibinfo {volume} {16}},\
  \bibinfo {pages} {177–182} (\bibinfo {year} {1991})}\BibitemShut {NoStop}%
\bibitem [{\citenamefont {Huttner}\ and\ \citenamefont
  {Barnett}(1992)}]{Huttner_1992}%
  \BibitemOpen
  \bibfield  {author} {\bibinfo {author} {\bibfnamefont {B.}~\bibnamefont
  {Huttner}}\ and\ \bibinfo {author} {\bibfnamefont {S.~M.}\ \bibnamefont
  {Barnett}},\ }\bibfield  {title} {\bibinfo {title} {Quantization of the
  electromagnetic field in dielectrics},\ }\href
  {https://doi.org/10.1103/physreva.46.4306} {\bibfield  {journal} {\bibinfo
  {journal} {Physical Review A}\ }\textbf {\bibinfo {volume} {46}},\ \bibinfo
  {pages} {4306} (\bibinfo {year} {1992})}\BibitemShut {NoStop}%
\bibitem [{\citenamefont {Gruner}\ and\ \citenamefont
  {Welsch}(1996)}]{Gruner_1996}%
  \BibitemOpen
  \bibfield  {author} {\bibinfo {author} {\bibfnamefont {T.}~\bibnamefont
  {Gruner}}\ and\ \bibinfo {author} {\bibfnamefont {D.-G.}\ \bibnamefont
  {Welsch}},\ }\bibfield  {title} {\bibinfo {title} {{G}reen-function approach
  to the radiation-field quantization for homogeneous and inhomogeneous
  {K}ramers-{K}ronig dielectrics},\ }\href
  {https://doi.org/10.1103/physreva.53.1818} {\bibfield  {journal} {\bibinfo
  {journal} {Physical Review A}\ }\textbf {\bibinfo {volume} {53}},\ \bibinfo
  {pages} {1818} (\bibinfo {year} {1996})}\BibitemShut {NoStop}%
\bibitem [{\citenamefont {Bhat}\ and\ \citenamefont {Sipe}(2006)}]{Bhat_2006}%
  \BibitemOpen
  \bibfield  {author} {\bibinfo {author} {\bibfnamefont {N.~A.~R.}\
  \bibnamefont {Bhat}}\ and\ \bibinfo {author} {\bibfnamefont {J.~E.}\
  \bibnamefont {Sipe}},\ }\bibfield  {title} {\bibinfo {title} {Hamiltonian
  treatment of the electromagnetic field in dispersive and absorptive
  structured media},\ }\href {https://doi.org/10.1103/physreva.73.063808}
  {\bibfield  {journal} {\bibinfo  {journal} {Physical Review A}\ }\textbf
  {\bibinfo {volume} {73}},\ \bibinfo {pages} {063808} (\bibinfo {year}
  {2006})}\BibitemShut {NoStop}%
\bibitem [{\citenamefont {Raabe}\ \emph {et~al.}(2007)\citenamefont {Raabe},
  \citenamefont {Scheel},\ and\ \citenamefont {Welsch}}]{Raabe_2007}%
  \BibitemOpen
  \bibfield  {author} {\bibinfo {author} {\bibfnamefont {C.}~\bibnamefont
  {Raabe}}, \bibinfo {author} {\bibfnamefont {S.}~\bibnamefont {Scheel}},\ and\
  \bibinfo {author} {\bibfnamefont {D.-G.}\ \bibnamefont {Welsch}},\ }\bibfield
   {title} {\bibinfo {title} {Unified approach to {QED} in arbitrary linear
  media},\ }\href {https://doi.org/10.1103/physreva.75.053813} {\bibfield
  {journal} {\bibinfo  {journal} {Physical Review A}\ }\textbf {\bibinfo
  {volume} {75}},\ \bibinfo {pages} {053813} (\bibinfo {year}
  {2007})}\BibitemShut {NoStop}%
\bibitem [{\citenamefont {Hodgson}\ \emph {et~al.}(2022)\citenamefont
  {Hodgson}, \citenamefont {Southall}, \citenamefont {Purdy},\ and\
  \citenamefont {Beige}}]{Hodgson_2022}%
  \BibitemOpen
  \bibfield  {author} {\bibinfo {author} {\bibfnamefont {D.}~\bibnamefont
  {Hodgson}}, \bibinfo {author} {\bibfnamefont {J.}~\bibnamefont {Southall}},
  \bibinfo {author} {\bibfnamefont {R.}~\bibnamefont {Purdy}},\ and\ \bibinfo
  {author} {\bibfnamefont {A.}~\bibnamefont {Beige}},\ }\bibfield  {title}
  {\bibinfo {title} {Local photons},\ }\href
  {https://doi.org/10.3389/fphot.2022.978855} {\bibfield  {journal} {\bibinfo
  {journal} {Frontiers in Photonics}\ }\textbf {\bibinfo {volume} {3}},\
  \bibinfo {pages} {978855} (\bibinfo {year} {2022})}\BibitemShut {NoStop}%
\bibitem [{\citenamefont {Waite}\ \emph {et~al.}(2025)\citenamefont {Waite},
  \citenamefont {Hodgson}, \citenamefont {Lang}, \citenamefont {Alapatt},\ and\
  \citenamefont {Beige}}]{Waite_2025}%
  \BibitemOpen
  \bibfield  {author} {\bibinfo {author} {\bibfnamefont {G.}~\bibnamefont
  {Waite}}, \bibinfo {author} {\bibfnamefont {D.}~\bibnamefont {Hodgson}},
  \bibinfo {author} {\bibfnamefont {B.}~\bibnamefont {Lang}}, \bibinfo {author}
  {\bibfnamefont {V.}~\bibnamefont {Alapatt}},\ and\ \bibinfo {author}
  {\bibfnamefont {A.}~\bibnamefont {Beige}},\ }\bibfield  {title} {\bibinfo
  {title} {Local-photon model of the momentum of light},\ }\href
  {https://doi.org/10.1103/physreva.111.023703} {\bibfield  {journal} {\bibinfo
   {journal} {Physical Review A}\ }\textbf {\bibinfo {volume} {111}},\ \bibinfo
  {pages} {023703} (\bibinfo {year} {2025})}\BibitemShut {NoStop}%
\bibitem [{\citenamefont {Alekseev}\ and\ \citenamefont
  {P.}(1966)}]{Alekseev_1966}%
  \BibitemOpen
  \bibfield  {author} {\bibinfo {author} {\bibfnamefont {A.~I.}\ \bibnamefont
  {Alekseev}}\ and\ \bibinfo {author} {\bibfnamefont {N.~Y.}\ \bibnamefont
  {P.}},\ }\bibfield  {title} {\bibinfo {title} {Quantisation of the
  electromagnetic field in a dispersive medium},\ }\href@noop {} {\bibfield
  {journal} {\bibinfo  {journal} {Sov. Phys. JETP}\ }\textbf {\bibinfo {volume}
  {23}},\ \bibinfo {pages} {608} (\bibinfo {year} {1966})}\BibitemShut
  {NoStop}%
\bibitem [{\citenamefont {Klyshko}(1988)}]{klyshko1988book}%
  \BibitemOpen
  \bibfield  {author} {\bibinfo {author} {\bibfnamefont {D.}~\bibnamefont
  {Klyshko}},\ }\href {https://doi.org/10.1201/9780203743508} {\emph {\bibinfo
  {title} {Photons and Nonlinear Optics}}}\ (\bibinfo  {publisher}
  {Routledge},\ \bibinfo {year} {1988})\BibitemShut {NoStop}%
\bibitem [{\citenamefont {Messinger}\ \emph {et~al.}(2020)\citenamefont
  {Messinger}, \citenamefont {Westerberg},\ and\ \citenamefont
  {Barnett}}]{Messinger_2020}%
  \BibitemOpen
  \bibfield  {author} {\bibinfo {author} {\bibfnamefont {A.}~\bibnamefont
  {Messinger}}, \bibinfo {author} {\bibfnamefont {N.}~\bibnamefont
  {Westerberg}},\ and\ \bibinfo {author} {\bibfnamefont {S.~M.}\ \bibnamefont
  {Barnett}},\ }\bibfield  {title} {\bibinfo {title} {Spontaneous emission in
  anisotropic dielectrics},\ }\href
  {https://doi.org/10.1103/physreva.102.013721} {\bibfield  {journal} {\bibinfo
   {journal} {Physical Review A}\ }\textbf {\bibinfo {volume} {102}},\ \bibinfo
  {pages} {013721} (\bibinfo {year} {2020})}\BibitemShut {NoStop}%
\bibitem [{\citenamefont {Quesada}\ \emph {et~al.}(2020)\citenamefont
  {Quesada}, \citenamefont {Triginer}, \citenamefont {Vidrighin},\ and\
  \citenamefont {Sipe}}]{Quesada_2020}%
  \BibitemOpen
  \bibfield  {author} {\bibinfo {author} {\bibfnamefont {N.}~\bibnamefont
  {Quesada}}, \bibinfo {author} {\bibfnamefont {G.}~\bibnamefont {Triginer}},
  \bibinfo {author} {\bibfnamefont {M.~D.}\ \bibnamefont {Vidrighin}},\ and\
  \bibinfo {author} {\bibfnamefont {J.~E.}\ \bibnamefont {Sipe}},\ }\bibfield
  {title} {\bibinfo {title} {Theory of high-gain twin-beam generation in
  waveguides: From {M}axwell{\textquotesingle}s equations to efficient
  simulation},\ }\href {https://doi.org/10.1103/physreva.102.033519} {\bibfield
   {journal} {\bibinfo  {journal} {Physical Review A}\ }\textbf {\bibinfo
  {volume} {102}},\ \bibinfo {pages} {033519} (\bibinfo {year}
  {2020})}\BibitemShut {NoStop}%
\bibitem [{\citenamefont {Abram}(1987)}]{Abram_1987}%
  \BibitemOpen
  \bibfield  {author} {\bibinfo {author} {\bibfnamefont {I.}~\bibnamefont
  {Abram}},\ }\bibfield  {title} {\bibinfo {title} {Quantum theory of light
  propagation: Linear medium},\ }\href
  {https://doi.org/10.1103/PhysRevA.35.4661} {\bibfield  {journal} {\bibinfo
  {journal} {Phys. Rev. A}\ }\textbf {\bibinfo {volume} {35}},\ \bibinfo
  {pages} {4661} (\bibinfo {year} {1987})}\BibitemShut {NoStop}%
\bibitem [{\citenamefont {Huttner}\ \emph {et~al.}(1990)\citenamefont
  {Huttner}, \citenamefont {Serulnik},\ and\ \citenamefont
  {Ben-Aryeh}}]{Huttner_1990}%
  \BibitemOpen
  \bibfield  {author} {\bibinfo {author} {\bibfnamefont {B.}~\bibnamefont
  {Huttner}}, \bibinfo {author} {\bibfnamefont {S.}~\bibnamefont {Serulnik}},\
  and\ \bibinfo {author} {\bibfnamefont {Y.}~\bibnamefont {Ben-Aryeh}},\
  }\bibfield  {title} {\bibinfo {title} {Quantum analysis of light propagation
  in a parametric amplifier},\ }\href
  {https://doi.org/10.1103/physreva.42.5594} {\bibfield  {journal} {\bibinfo
  {journal} {Physical Review A}\ }\textbf {\bibinfo {volume} {42}},\ \bibinfo
  {pages} {5594} (\bibinfo {year} {1990})}\BibitemShut {NoStop}%
\bibitem [{\citenamefont {Ben-Aryeh}\ and\ \citenamefont
  {Serulnik}(1991)}]{Ben_Aryeh_1991}%
  \BibitemOpen
  \bibfield  {author} {\bibinfo {author} {\bibfnamefont {Y.}~\bibnamefont
  {Ben-Aryeh}}\ and\ \bibinfo {author} {\bibfnamefont {S.}~\bibnamefont
  {Serulnik}},\ }\bibfield  {title} {\bibinfo {title} {The quantum treatment of
  propagation in non-linear optical media by the use of temporal modes},\
  }\href {https://doi.org/10.1016/0375-9601(91)90650-w} {\bibfield  {journal}
  {\bibinfo  {journal} {Physics Letters A}\ }\textbf {\bibinfo {volume}
  {155}},\ \bibinfo {pages} {473–479} (\bibinfo {year} {1991})}\BibitemShut
  {NoStop}%
\bibitem [{\citenamefont {Perina}\ and\ \citenamefont
  {Perina}(1995)}]{Perina_1995}%
  \BibitemOpen
  \bibfield  {author} {\bibinfo {author} {\bibfnamefont {J.}~\bibnamefont
  {Perina}}\ and\ \bibinfo {author} {\bibfnamefont {J.}~\bibnamefont
  {Perina}},\ }\bibfield  {title} {\bibinfo {title} {Photon statistics of a
  contradirectional nonlinear coupler},\ }\href
  {https://doi.org/10.1088/1355-5111/7/5/007} {\bibfield  {journal} {\bibinfo
  {journal} {Quantum and Semiclassical Optics: Journal of the European Optical
  Society Part B}\ }\textbf {\bibinfo {volume} {7}},\ \bibinfo {pages} {849}
  (\bibinfo {year} {1995})}\BibitemShut {NoStop}%
\bibitem [{\citenamefont {Horoshko}(2022)}]{Horoshko_2022}%
  \BibitemOpen
  \bibfield  {author} {\bibinfo {author} {\bibfnamefont {D.~B.}\ \bibnamefont
  {Horoshko}},\ }\bibfield  {title} {\bibinfo {title} {Generator of spatial
  evolution of the electromagnetic field},\ }\href
  {https://doi.org/10.1103/PhysRevA.105.013708} {\bibfield  {journal} {\bibinfo
   {journal} {Phys. Rev. A}\ }\textbf {\bibinfo {volume} {105}},\ \bibinfo
  {pages} {013708} (\bibinfo {year} {2022})}\BibitemShut {NoStop}%
\bibitem [{\citenamefont {Liñares}\ and\ \citenamefont
  {Nistal}(2003)}]{Li_ares_2003}%
  \BibitemOpen
  \bibfield  {author} {\bibinfo {author} {\bibfnamefont {J.}~\bibnamefont
  {Liñares}}\ and\ \bibinfo {author} {\bibfnamefont {M.~C.}\ \bibnamefont
  {Nistal}},\ }\bibfield  {title} {\bibinfo {title} {Quantization of coupled
  modes propagation in integrated optical waveguides},\ }\href
  {https://doi.org/10.1080/09500340308235185} {\bibfield  {journal} {\bibinfo
  {journal} {Journal of Modern Optics}\ }\textbf {\bibinfo {volume} {50}},\
  \bibinfo {pages} {781–790} (\bibinfo {year} {2003})}\BibitemShut {NoStop}%
\bibitem [{\citenamefont {Liñares}\ \emph {et~al.}(2008)\citenamefont
  {Liñares}, \citenamefont {Nistal},\ and\ \citenamefont
  {Barral}}]{Li_ares_2008}%
  \BibitemOpen
  \bibfield  {author} {\bibinfo {author} {\bibfnamefont {J.}~\bibnamefont
  {Liñares}}, \bibinfo {author} {\bibfnamefont {M.~C.}\ \bibnamefont
  {Nistal}},\ and\ \bibinfo {author} {\bibfnamefont {D.}~\bibnamefont
  {Barral}},\ }\bibfield  {title} {\bibinfo {title} {Quantization of coupled 1d
  vector modes in integrated photonic waveguides},\ }\href
  {https://doi.org/10.1088/1367-2630/10/6/063023} {\bibfield  {journal}
  {\bibinfo  {journal} {New Journal of Physics}\ }\textbf {\bibinfo {volume}
  {10}},\ \bibinfo {pages} {063023} (\bibinfo {year} {2008})}\BibitemShut
  {NoStop}%
\bibitem [{\citenamefont {Brecht}\ \emph {et~al.}(2015)\citenamefont {Brecht},
  \citenamefont {Reddy}, \citenamefont {Silberhorn},\ and\ \citenamefont
  {Raymer}}]{Brecht_2015}%
  \BibitemOpen
  \bibfield  {author} {\bibinfo {author} {\bibfnamefont {B.}~\bibnamefont
  {Brecht}}, \bibinfo {author} {\bibfnamefont {D.~V.}\ \bibnamefont {Reddy}},
  \bibinfo {author} {\bibfnamefont {C.}~\bibnamefont {Silberhorn}},\ and\
  \bibinfo {author} {\bibfnamefont {M.}~\bibnamefont {Raymer}},\ }\bibfield
  {title} {\bibinfo {title} {Photon temporal modes: A complete framework for
  quantum information science},\ }\href
  {https://doi.org/10.1103/physrevx.5.041017} {\bibfield  {journal} {\bibinfo
  {journal} {Physical Review X}\ }\textbf {\bibinfo {volume} {5}},\ \bibinfo
  {pages} {041017} (\bibinfo {year} {2015})}\BibitemShut {NoStop}%
\bibitem [{\citenamefont {Raymer}\ and\ \citenamefont
  {Walmsley}(2020)}]{Raymer_2020}%
  \BibitemOpen
  \bibfield  {author} {\bibinfo {author} {\bibfnamefont {M.~G.}\ \bibnamefont
  {Raymer}}\ and\ \bibinfo {author} {\bibfnamefont {I.~A.}\ \bibnamefont
  {Walmsley}},\ }\bibfield  {title} {\bibinfo {title} {Temporal modes in
  quantum optics: then and now},\ }\href
  {https://doi.org/10.1088/1402-4896/ab6153} {\bibfield  {journal} {\bibinfo
  {journal} {Physica Scripta}\ }\textbf {\bibinfo {volume} {95}},\ \bibinfo
  {pages} {064002} (\bibinfo {year} {2020})}\BibitemShut {NoStop}%
\bibitem [{\citenamefont {Ebers}\ \emph {et~al.}(2023)\citenamefont {Ebers},
  \citenamefont {Ferreri}, \citenamefont {Hammer}, \citenamefont {Albert},
  \citenamefont {Meier}, \citenamefont {F\"{o}rstner},\ and\ \citenamefont
  {Sharapova}}]{EFH_2023}%
  \BibitemOpen
  \bibfield  {author} {\bibinfo {author} {\bibfnamefont {L.}~\bibnamefont
  {Ebers}}, \bibinfo {author} {\bibfnamefont {A.}~\bibnamefont {Ferreri}},
  \bibinfo {author} {\bibfnamefont {M.}~\bibnamefont {Hammer}}, \bibinfo
  {author} {\bibfnamefont {M.}~\bibnamefont {Albert}}, \bibinfo {author}
  {\bibfnamefont {C.}~\bibnamefont {Meier}}, \bibinfo {author} {\bibfnamefont
  {J.}~\bibnamefont {F\"{o}rstner}},\ and\ \bibinfo {author} {\bibfnamefont
  {P.~R.}\ \bibnamefont {Sharapova}},\ }\bibfield  {title} {\bibinfo {title}
  {Corrigendum: Flexible source of correlated photons based on {LNOI} rib
  waveguides {(2022 J.\ Phys.\ Photonics~4~025001)}},\ }\href@noop {}
  {\bibfield  {journal} {\bibinfo  {journal} {Journal of Physics: Photonics}\
  }\textbf {\bibinfo {volume} {5}},\ \bibinfo {pages} {029501} (\bibinfo {year}
  {2023})}\BibitemShut {NoStop}%
\bibitem [{\citenamefont {Herzinger}\ \emph {et~al.}(1998)\citenamefont
  {Herzinger}, \citenamefont {Johs}, \citenamefont {McGahan}, \citenamefont
  {Woollam},\ and\ \citenamefont {Paulson}}]{Herzinger_1998}%
  \BibitemOpen
  \bibfield  {author} {\bibinfo {author} {\bibfnamefont {C.~M.}\ \bibnamefont
  {Herzinger}}, \bibinfo {author} {\bibfnamefont {B.}~\bibnamefont {Johs}},
  \bibinfo {author} {\bibfnamefont {W.~A.}\ \bibnamefont {McGahan}}, \bibinfo
  {author} {\bibfnamefont {J.~A.}\ \bibnamefont {Woollam}},\ and\ \bibinfo
  {author} {\bibfnamefont {W.}~\bibnamefont {Paulson}},\ }\bibfield  {title}
  {\bibinfo {title} {Ellipsometric determination of optical constants for
  silicon and thermally grown silicon dioxide via a multi-sample,
  multi-wavelength, multi-angle investigation},\ }\href@noop {} {\bibfield
  {journal} {\bibinfo  {journal} {Journal of Applied Physics}\ }\textbf
  {\bibinfo {volume} {83}},\ \bibinfo {pages} {3323} (\bibinfo {year}
  {1998})}\BibitemShut {NoStop}%
\bibitem [{JCMwave()}]{JCMwave}%
  \BibitemOpen
  JCMwave,\ \href@noop {} {}\bibinfo {note} {JCMwave GmbH, Berlin, Germany;
  \href{https://www.jcmwave.com/}{https://www.jcmwave.com}}\BibitemShut
  {NoStop}%
\bibitem [{\citenamefont {Hammer}\ \emph {et~al.}(2024)\citenamefont {Hammer},
  \citenamefont {Babel}, \citenamefont {Farheen}, \citenamefont {Padberg},
  \citenamefont {Scheytt}, \citenamefont {Silberhorn},\ and\ \citenamefont
  {Förstner}}]{Hammer_2024}%
  \BibitemOpen
  \bibfield  {author} {\bibinfo {author} {\bibfnamefont {M.}~\bibnamefont
  {Hammer}}, \bibinfo {author} {\bibfnamefont {S.}~\bibnamefont {Babel}},
  \bibinfo {author} {\bibfnamefont {H.}~\bibnamefont {Farheen}}, \bibinfo
  {author} {\bibfnamefont {L.}~\bibnamefont {Padberg}}, \bibinfo {author}
  {\bibfnamefont {J.~C.}\ \bibnamefont {Scheytt}}, \bibinfo {author}
  {\bibfnamefont {C.}~\bibnamefont {Silberhorn}},\ and\ \bibinfo {author}
  {\bibfnamefont {J.}~\bibnamefont {Förstner}},\ }\bibfield  {title} {\bibinfo
  {title} {Estimation of losses caused by sidewall roughness in thin-film
  lithium niobate rib and strip waveguides},\ }\href
  {https://doi.org/10.1364/oe.521766} {\bibfield  {journal} {\bibinfo
  {journal} {Optics Express}\ }\textbf {\bibinfo {volume} {32}},\ \bibinfo
  {pages} {22878} (\bibinfo {year} {2024})}\BibitemShut {NoStop}%
\bibitem [{\citenamefont {Kopylov}\ \emph {et~al.}(2025)\citenamefont
  {Kopylov}, \citenamefont {Meier},\ and\ \citenamefont
  {Sharapova}}]{Kopylov_2025}%
  \BibitemOpen
  \bibfield  {author} {\bibinfo {author} {\bibfnamefont {D.~A.}\ \bibnamefont
  {Kopylov}}, \bibinfo {author} {\bibfnamefont {T.}~\bibnamefont {Meier}},\
  and\ \bibinfo {author} {\bibfnamefont {P.~R.}\ \bibnamefont {Sharapova}},\
  }\bibfield  {title} {\bibinfo {title} {Theory of multimode squeezed light
  generation in lossy media},\ }\href
  {https://doi.org/10.22331/q-2025-02-04-1621} {\bibfield  {journal} {\bibinfo
  {journal} {Quantum}\ }\textbf {\bibinfo {volume} {9}},\ \bibinfo {pages}
  {1621} (\bibinfo {year} {2025})}\BibitemShut {NoStop}%
\end{thebibliography}%

\newpage

\appendix



\section{ Mathematical identities}
\label{appendix_1}

\noindent
Here, we collect a few mathematical results that will be used for the later proof of the orthogonality conditions.
\begin{lemma}
    For a transverse-evanescent field $\vec{X}(\vec{r})$ with $\lim_{|\vec{r}_\perp|\to\infty} \vec{X}(\vec{r})=0$, the following equality holds:
    \begin{equation}
        \int_{-\infty}^\infty d\vec{r}_\perp \; \nabla \cdot \vec{X}(\vec{r}) = \dfrac{\partial}{\partial z} \int_{-\infty}^\infty d\vec{r}_\perp \left[   X_z(\vec{r}) \right],
    \end{equation}
    where $\vec{r}_\perp \equiv (x,y)$, $ d\vec{r}_\perp  \equiv dxdy$ and $X_z(\vec{r})\equiv \vec{X}(\vec{r}) \cdot \vec{v}_z $.
    \label{lemma_1}
\end{lemma}
\begin{proof}
    Let us represent the operator $\nabla$ as $\nabla = \nabla_\perp + \vec{v}_z \dfrac{\partial}{\partial z} $, where $\nabla_\perp \equiv \vec{v}_x \dfrac{\partial}{\partial x} + \vec{v}_y \dfrac{\partial}{\partial y}$ and $\vec{v}_i$ are the unit vectors.
   The integral over a surface $\Sigma$ in the $xy$-plane reads
    \begin{equation}
        \int_\Sigma d\vec{r}_\perp [\nabla \cdot \vec{X}(\vec{r})] = \int_\Sigma d\vec{r}_\perp [\nabla_\perp \cdot \vec{X}(\vec{r})]  + \int_\Sigma d\vec{r}_\perp  \dfrac{\partial}{\partial z} X_z(\vec{r}).
    \end{equation}
    Applying the two-dimensional divergence theorem~\cite{Yariv_Yeh_book}, the second integral is
    \begin{equation}
        \int_\Sigma d\vec{r}_\perp \; \nabla_\perp \cdot \vec{X}(\vec{r})  = \oint_{\partial\Sigma} \vec{X}(\vec{r}) \cdot \vec{v}_z dl,
    \end{equation}
    where the $\partial\Sigma$ is the boundary of the surface $\Sigma$.
    Taking as $\Sigma$ the full $xy$-plane, for the transverse-evanescent field, the right-hand side integral is equal to zero.

\end{proof}
\begin{lemma}
For two complex fields $\vec{X}_{1,2}(\vec{r})$ and an Hermitian tensor $\tensor{\kappa}(\vec{r})$     
the following equalities hold
     \begin{equation}
         \vec{X}_1 \cdot (\tensor{\kappa}\vec{X}_2)^* = (\tensor{\kappa}\vec{X}_1) \cdot \vec{X}^*_2,
     \end{equation}
\vspace{-4ex}
     \begin{equation}
         \mathrm{Im} (\vec{X}_1 \cdot (\tensor{\kappa}\vec{X}_1)^*) = 0.
     \end{equation}
     \label{lemma_2}
These are a direct consequence of $\tensor{\kappa}$ being Hermitian.

\end{lemma}
%
%
\vspace{2ex}
\paragraph*{A vector identity.}
\label{sub_vector_identity}
\noindent
In order to render the further derivations a little more compact, 
let us introduce vector 
field valued operators
$\vec{R}$, $\vec{Q}_\alpha$, $\alpha \in {x, y, z}$, and $\vec{W}$, defined for fields $\vec{X}(\vec{r})$ and $\vec{Y}(\vec{r})$ as follows.
First we set
\begin{equation}
    \vec{R} (\vec{X}, \vec{Y}) \equiv \vec{X}  (\nabla \cdot \vec{Y})- \vec{Y}\times (\nabla \times \vec{X}).
    \label{seq_vector_R}
\end{equation}
Next, the recipe
\begin{equation}
    \vec{Q}_{\alpha}(\vec{X}, \vec{Y}) \equiv \big(Q_{\alpha x}(\vec{X}, \vec{Y}), Q_{\alpha y}(\vec{X}, \vec{Y}), Q_{\alpha z}(\vec{X}, \vec{Y}) \big),
\end{equation}
where the components of the tensor $\tensor{Q}$ are
\begin{equation}
    Q_{\alpha\beta}(\vec{X}, \vec{Y}) = X_\alpha Y_\beta - \frac12 \vec{X} \cdot \vec{Y}\delta_{\alpha\beta},
    \label{seq_tensor_Q}
\end{equation}
for $\alpha, \beta \in \set{x,y,z}$. 
Note that $\nabla \cdot \vec{Q}_\alpha(\vec{X}, \vec{Y}) \equiv \sum_\beta \frac{\partial}{\partial \beta} Q_{\alpha\beta}(\vec{X}, \vec{Y})$.
Last, the components of $\vec{W}$ are defined as 
\begin{equation}
    W_\alpha (\vec{X}, \vec{Y}) = \frac12 \sum_{\beta \in \set{ x,y,z  }} \left( X_\beta \dfrac{\partial}{\partial \alpha} Y_\beta - Y_\beta \dfrac{\partial}{\partial \alpha} X_\beta \right).
    \label{seq_vector_V}
\end{equation}
Then, the following lemma holds:
\begin{lemma}
    \begin{equation}
        R_\alpha (\vec{X}, \vec{Y}) = \nabla \cdot \vec{Q}_\alpha(\vec{X}, \vec{Y}) + W_\alpha(\vec{X}, \vec{Y}).
\label{lemma3eq}
    \end{equation}
    \label{lemma_3}
\end{lemma}
\begin{proof}
    Let us take $\alpha=x$ and consider the $x$-component of $\vec{R}$,
    \begin{multline}
        R_x =  X_x \frac{\partial Y_x}{\partial x}  +  \underbrace{X_x \frac{\partial Y_y}{\partial y}}_1  +  \underbrace{X_x \frac{\partial Y_z}{\partial z}}_2    -  Y_y\frac{\partial X_y}{\partial x} \\ +   \underbrace{Y_y \frac{\partial X_x}{\partial y}}_1  - Y_z \frac{\partial X_z}{\partial x}  + \underbrace{Y_z \frac{\partial  X_x}{\partial z}}_2.
    \end{multline}
    After an addition and a subtraction of the term $\frac{\partial}{\partial x} (X_x Y_x)$, and combining the terms `$1$' and `$2$', we obtain
    \begin{multline}
        R_x = \frac{\partial}{\partial x} (X_x Y_x)  + \frac{\partial}{\partial y} (Y_y X_x)    +   \frac{\partial}{\partial z} (X_x Y_z)  \\ - \frac{\partial}{\partial x} (X_x Y_x)  +    X_x \frac{\partial Y_x}{\partial x} -  Y_y\frac{\partial X_y}{\partial x}  - Y_z \frac{\partial X_z}{\partial x} .
    \end{multline}
    Then, the identity
    \begin{align}
        - Y  \frac{\partial X}{\partial x} & = X \frac{\partial Y}{\partial x}   - \frac{\partial}{\partial x} (X Y), 
    \end{align}
    when applied to the last two terms, leads to 
    \begin{multline}
        R_x =        \Big[\frac{\partial}{\partial x} (X_x Y_x)+\frac{\partial}{\partial y} (X_x Y_y) +  \frac{\partial}{\partial z} (X_x Y_z) \Big] \\
        - \Big[ \frac{\partial}{\partial x} (X_x Y_x) + \frac{\partial}{\partial x} (X_y Y_y) + \frac{\partial}{\partial x} (X_z Y_z)    \Big]    \\
        +   X_x \frac{\partial Y_x}{\partial x} + X_y \frac{\partial Y_y}{\partial x}    + X_z \frac{\partial Y_z}{\partial x}.
        \label{seq_Rx_intermediate}
    \end{multline}
    After splitting the second bracket into two equal parts, and recombining one of them with the terms from the third line, $R_x$ can be written as
    \begin{multline}
        R_x = \Big[
            \frac{\partial}{\partial x} (X_x Y_x)
            + \frac{\partial}{\partial y} (X_x Y_y)
            + \frac{\partial}{\partial z} (X_x Y_z)
            \Big] \\
        - \frac12\Big[
            \frac{\partial}{\partial x} (X_x Y_x)
            + \frac{\partial}{\partial x} (X_y Y_y)
            + \frac{\partial}{\partial x} (X_z Y_z)
            \Big] \\
        + \frac12\Big[
            X_x \frac{\partial Y_x}{\partial x} - Y_x \frac{\partial X_x}{\partial x}
            + X_y \frac{\partial Y_y}{\partial x} - Y_y \frac{\partial X_y}{\partial x} \\
            + X_z \frac{\partial Y_z}{\partial x} - Y_z \frac{\partial X_z}{\partial x}
            \Big].
        \label{seq_approx_buf_1}
    \end{multline}
    The first two brackets in Eq.~\eqref{seq_approx_buf_1} give us
    \begin{equation}
        \sum_{\beta \in \set{x,y,z}} \frac{\partial}{\partial \beta}  \Big( X_x Y_\beta - \frac12 \vec{X} \cdot \vec{Y}\delta_{x\beta} \Big) 
= 
\nabla \cdot \vec{Q}_x(\vec{X}, \vec{Y}),
    \end{equation}
    while the third bracket in Eq.~\eqref{seq_approx_buf_1} is
    \begin{multline}
        \frac12 \sum_{\beta \in \set{x,y,z}} \Big( X_\beta \dfrac{\partial}{\partial x} Y_\beta - Y_\beta \dfrac{\partial}{\partial x} X_\beta \Big) 
= 
W_x(\vec{X}, \vec{Y}).
    \end{multline}
    This finishes the proof for the component $R_x$ of Eq.~\eqref{lemma3eq}.
    Respective results for $R_y$ and $R_z$ are derived analogously.
\end{proof}


\section{Maxwell equations in the frequency domain for straight dielectric waveguides}

\noindent
In our study, we consider the macroscopic electromagnetic fields in a medium in the absence of free charges 
and free currents.
In addition, we restrict ourselves to linear, lossless, 
non-magnetic,
and potentially anisotropic dielectrics.
As a result, the corresponding Maxwell equations in the frequency domain are
\begin{subequations}
    \label{seq_maxwell_freq_WG_full}
    \begin{align}
        \label{seq_maxwell_freq_WG_1}
        \nabla \times \vec{e} & =  i\omega  \vec{b} ,
        \\
        \label{seq_maxwell_freq_WG_2}
        \nabla \times \vec{h} & = -i\omega \vec{d},
        \\
        \label{seq_maxwell_freq_WG_3}
        \nabla \cdot \vec{d}  & = 0,
        \\
        \label{seq_maxwell_freq_WG_4}
        \nabla \cdot \vec{b}  & = 0, 
    \end{align}
\end{subequations}
with the material equations in the form
\begin{subequations}
    \label{seq_material_freq_full}
    \begin{align}
        \label{seq_material_freq_5}
        \vec{d} & = \epsilon_0 \tensor{\varepsilon} \cdot \vec{e}, \\
        \label{seq_material_freq_6}
        \vec{b} & = \mu_0 \vec{h}.
    \end{align}
\end{subequations}
The electromagnetic fields 
$\vec{e}$, 
$\vec{h}$, 
$\vec{d}$, 
$\vec{b}$, 
depend on space $\vec{r}$ and frequency $\omega$. 
The dielectric profile of the channel is given through the 
relative permittivity $\tensor{\varepsilon}(\vec{r}_\perp, \omega)$ that varies with the cross sectional coordinates $\vec{r}_\perp$ only.
For lossless media, this tensor 
is Hermitian, $\tensor{\varepsilon}=\tensor{\varepsilon}^\dagger$, where the symbol $.^\dagger$ denotes the adjoint.
Therefore, Lemma~\ref{lemma_2} is applicable to the electric fields $\vec{d}$ and $\vec{e}$.


\section{Derivation of orthogonality conditions}
\label{app_derivation}
\noindent
Consider a set of solutions of Eqs.~\eqref{seq_maxwell_freq_WG_full} and \eqref{seq_material_freq_full}, for the same waveguide structure, i.e.\ for the same permittivity $\tensor{\varepsilon}$, in modal form 
\begin{equation}
\label{mode_appendix}
(\vec{e}, \vec{h}, \vec{d}, \vec{b})(\vec{r}) = \big(\mathbfcal{E}, \mathbfcal{H}, \mathbfcal{D}, \mathbfcal{B} \big)(\vec{r}_\perp)\, e^{i k z} .
\end{equation}
Specifically, we pick two such solutions, one with index $n$, associated with fields 
$\vec{e}_n$, $\vec{h}_n$, $\vec{d}_n$, $\vec{b}_n$, 
mode profiles 
$\mathbfcal{E}_n$, $\mathbfcal{H}_n$, $\mathbfcal{D}_n$, $\mathbfcal{B}_n$, and propagation constant $k_n$, the other with index $m$. We assume that the profiles are integrable over the channel cross section domain (guided modes), and that these modes are nodegenerate, i.e., that 
$k_n \neq k_m$ for $n\neq m$.


\subsection{Energy orthogonality condition}
\label{proof_for_energy_orthogonality_condition}

\noindent
The proof of Eq.~\eqref{eq_orthogonality_condition_energy} can be found, e.g., in Refs.~\cite{Vassallo_1991_book,Jackson_book,Yariv_Yeh_book}; however, for reasons of self-consistency, we briefly recall it here. 
Based on the two modal solutions $n$, $m$, we consider the quantity
    \begin{equation}
        w_{nm} = \nabla \cdot (\vec{e}^*_n \times \vec{h}_m + \vec{e}_m \times \vec{h}^*_n ).
    \end{equation}
    Using Eq.~\eqref{seq_maxwell_freq_WG_2} and applying the vector identity $\nabla \cdot (\vec{X} \times \vec{Y} ) = (\nabla \times \vec{X}) \cdot \vec{Y} -\vec{X} \cdot (\nabla \times \vec{Y})$, 
    we obtain
    \begin{equation}
        w_{nm} =  i \omega \Big[\vec{e}^*_n \cdot \vec{d}_m -\vec{e}_m \cdot  \vec{d}_n^* \Big].
    \end{equation}
    For Hermitian $\tensor{\varepsilon}$, Eq.~\eqref{seq_material_freq_5} and Lemma~\ref{lemma_2} give us $w_{nm}=0$, and, therefore,
    \begin{equation}
        \nabla \cdot (\vec{e}^*_n \times \vec{h}_m + \vec{e}_m \times \vec{h}^*_n ) = 0.
    \end{equation}
    Applying Lemma~\ref{lemma_1}, and substituting the solutions in their modal form \eqref{mode_appendix}, we obtain
    \begin{multline}
        i(k_m-k_n) e^{ i(k_m-k_n) z } \\ \cdot \int d\vec{r}_\perp \:  \big[\mathbfcal{E}^*_n \times \mathbfcal{H}_m + \mathbfcal{E}_m \times \mathbfcal{H}^*_n \big]_z = 0.
    \end{multline}
This proves the energy orthogonality condition \eqref{eq_orthogonality_condition_energy}.


\subsection{Momentum orthogonality condition}
\label{proof_for_momentum_orthogonality_condition}

\noindent
Here we present the derivation of 
Eq.~\eqref{eq_orthogonality_condition_momentum}.
For a pair of modal solutions $n$, $m$, let us construct a vector field
    \begin{equation}
        \vec{f}^{(1)} \equiv \vec{e}_n(\nabla \cdot \vec{d}^*_m) +  (\nabla \times \vec{h}^*_m  - i \omega \vec{d}^*_m)\times \vec{b}_n + \vec{h}^*_m (\nabla\cdot \vec{b}_n) \equiv 0
        \label{seq_vector_f1_}
    \end{equation}
$\vec{f}^{(1)}$ vanishes due to Eqs.~\eqref{seq_maxwell_freq_WG_2}, \eqref{seq_maxwell_freq_WG_3} and \eqref{seq_maxwell_freq_WG_4}.
    Substituting Eq.~\eqref{seq_maxwell_freq_WG_1} into~\eqref{seq_vector_f1_}, we obtain
    \begin{multline}
        \vec{f}^{(1)} = \vec{e}_n(\nabla \cdot \vec{d}^*_m) - \vec{d}^*_m \times ( \nabla \times \vec{e}_n)  \\ +  \vec{h}^*_m (\nabla\cdot \vec{b}_n)- \vec{b}_n \times (\nabla \times \vec{h}^*_m) .
    \end{multline}
The $z$-component of $\vec{f}^{(1)}$ can be written in terms of the product \eqref{seq_vector_R} as
    \begin{equation}
        f^{(1)}_z = \vec{R}_z (\vec{e}_n, \vec{d}^*_m) + \vec{R}_z (\vec{h}^*_m, \vec{b}_n) = 0 .
    \end{equation}
In a similar way, we define a second vector
    \begin{multline}
        \vec{f}^{(2)} \equiv \vec{e}^*_m (\nabla \cdot \vec{d}_n) - \vec{d}_n \times (\nabla \times \vec{e}^*_m)  \\ + \vec{h}_n (\nabla\cdot \vec{b}^*_m) - \vec{b}_m^* \times (\nabla \times \vec{h}_n) \equiv 0
    \end{multline}
    with the $z$-component
    \begin{equation}
        f^{(2)}_z = \vec{R}_z (\vec{e}^*_m, \vec{d}_n) + \vec{R}_z (\vec{h}_n, \vec{b}^*_m) = 0 .
    \end{equation}
Applying Lemma~\ref{lemma_3} to the sum $f^{(1)}_z+f^{(2)}_z=0$, we get
    \begin{multline}
        \nabla \cdot \Big[  \vec{Q}_z(\vec{e}_n, \vec{d}^*_m) +  \vec{Q}_z(\vec{h}^*_m, \vec{b}_n)  \\ + \vec{Q}_z(\vec{e}^*_m, \vec{d}_n) +  \vec{Q}_z( \vec{h}_n, \vec{b}^*_m)  \Big] \\
        + W_z(\vec{e}_n, \vec{d}^*_m)
        + W_z(\vec{e}^*_m, \vec{d}_n) = 0.
        \label{seq_app_divergence_1}
    \end{multline}
    Here we used Eq.~\eqref{seq_material_freq_6}, which gives us $W_z(\vec{h}^*_m, \vec{b}_n) + W_z(\vec{h}_n, \vec{b}^*_m)=0$.
By substituting the mode representations \eqref{mode_appendix}, we obtain, for the last two terms of the left hand side,
    \begin{multline}
        W_z(\vec{e}_n, \vec{d}^*_m) + W_z(\vec{e}^*_m, \vec{d}_n) =\frac12 e^{i(k_n - k_m)z} \cdot  \\ \cdot \big[
            i k_m  (\mathbfcal{D}_n \cdot \mathbfcal{E}^*_m   -  \mathbfcal{E}_n \cdot \mathbfcal{D}^*_m)  \\
            +  i k_n (\mathbfcal{E}^*_m \cdot \mathbfcal{D}_n   - \mathbfcal{D}^*_m \cdot  \mathbfcal{E}_n) \big] = 0.
            \label{seq_intermediate_step_WW}
    \end{multline}
In this Lemma~\ref{lemma_2} was used for the last equality.
As a result, Eq.~\eqref{seq_app_divergence_1} now reads
    \begin{multline}
        \nabla \cdot \big[  e^{i (k_n-k_m) z} \big(  \vec{Q}_z(\mathbfcal{E}_n, \mathbfcal{D}^*_m) +  \vec{Q}_z(\mathbfcal{H}^*_m, \mathbfcal{B}_n) + \\ \vec{Q}_z(\mathbfcal{E}^*_m, \mathbfcal{D}_n) +  \vec{Q}_z(\mathbfcal{H}_n, \mathbfcal{B}^*_m)  \big) \big] = 0.
        \label{seq_app_divergence_full}
    \end{multline}
After integrating over the transverse plane, and applying Lemma~\ref{lemma_1}, we arrive at
    \begin{multline}
        i (k_n-k_m) \int d\vec{r}_\perp \big[   Q_{zz}(\mathbfcal{E}_n, \mathbfcal{D}^*_m) +  Q_{zz}(\mathbfcal{H}^*_m, \mathbfcal{B}_n) + \\ Q_{zz}(\mathbfcal{E}^*_m, \mathbfcal{D}_n) +  Q_{zz}(\mathbfcal{H}_n, \mathbfcal{B}^*_m)  \big]  = 0.
   \end{multline}
This proves the momentum orthogonality condition Eq.~\eqref{eq_orthogonality_condition_momentum}.


\subsection{Connection between energy and momentum orthogonality conditions}
\label{proof_for_orthogonality_connection}
\noindent
It remains to derive the relation \eqref{eq_orthogonality_condition_connection} between the normalization constants that appear in the energy and momentum orthogonality conditions.
Let us consider one of the modes \eqref{mode_appendix} of our waveguide with index $n$. Substitution into Eqs.~\eqref{seq_maxwell_freq_WG_1} and \eqref{seq_maxwell_freq_WG_2} gives us
\begin{align}
    \nabla \times [\mathbfcal{E}_ne^{i k_n z}] & = i \omega  \mathbfcal{B}_ne^{i k_n z},
    \\
    \nabla \times [\mathbfcal{H}_ne^{i k_n z}] & = -i \omega  \mathbfcal{D}_ne^{i k_n z}.
\end{align}
Using the identities
\begin{align}
    \nabla \times [\mathbfcal{X}_ne^{i k_n z}] & = \Big[i k_n ( \vec{v}_z \times \mathbfcal{X}_n)  +  ( \nabla \times \mathbfcal{X}_n) \Big]e^{i k_n z},
    \\
    \nabla \cdot [\mathbfcal{X}_ne^{i k_n z}]    & =  \Big[ i k_n ( \vec{v}_z \cdot \mathbfcal{X}_n)  +  ( \nabla \cdot \mathbfcal{X}_n) \Big]  e^{i k_n z},
\end{align}
the equations for the profile fields take the form
%
\begin{subequations}
    \label{seq_max_WG_allcomp}
        \begin{align}
            i k_n [ \vec{v}_z \times \mathbfcal{E}_n]  +  [ \nabla \times \mathbfcal{E}_n]      & = i \omega  \mathbfcal{B}_n,
            \\
            i k_n [ \vec{v}_z \times \mathbfcal{H}_n]  +  [ \nabla \times \mathbfcal{H}_n]      & = -i \omega  \mathbfcal{D}_n,
        \end{align}
\end{subequations}
or
\begin{subequations}
    \label{seq_max_WG_allcomp_vz}
        \begin{align}
            \vec{v}_z\cdot[ \nabla \times \mathbfcal{E}_n]   & = i \omega  \vec{v}_z\cdot\mathbfcal{B}_n,
            \\
            \vec{v}_z\cdot[ \nabla \times \mathbfcal{H}_n]   & = -i \omega  \vec{v}_z\cdot\mathbfcal{D}_n.
        \end{align}
\end{subequations}
The momentum orthogonality condition \eqref{eq_orthogonality_condition_momentum} for the product of this $n$-th mode with itself is
\begin{equation}
    \int d\vec{r}_{\perp} ~ \Big[\mathbfcal{E}_n \odot  \mathbfcal{D}^*_n + \mathbfcal{E}^*_n \odot  \mathbfcal{D}_n + \mathbfcal{B}_n \odot  \mathbfcal{H}^*_n + \mathbfcal{B}^*_n \odot  \mathbfcal{H}_n \Big] = \xi\up{M}_n.
    \label{seq_orthog_intermed}
\end{equation}
By substituting Eqs.~\eqref{seq_max_WG_allcomp_vz},  we obtain the following terms for the contributions to the integrand of 
Eq.~\eqref{seq_orthog_intermed}: 
\begin{widetext}
    \begin{subequations}
    \label{seq_detailed}
        \begin{align}
            2\omega \: \mathbfcal{E}_n \odot  \mathbfcal{D}^*_n & =
            k_n  \mathbfcal{E}_n \cdot [ \mathbfcal{H}^*_n \times \vec{v}_z ]
            -i \mathbfcal{E}_n \cdot [ \nabla \times \mathbfcal{H}^*_n]
            +i 2 (\vec{v}_z \cdot \mathbfcal{E}_n) ( \vec{v}_z\cdot[ \nabla \times \mathbfcal{H}^*_n] ),
            \\
            2\omega \:\mathbfcal{E}^*_n \odot  \mathbfcal{D}_n  & =
            k_n \mathbfcal{E}^*_n \cdot [\mathbfcal{H}_n \times \vec{v}_z ]
            +i \mathbfcal{E}^*_n \cdot[ \nabla \times \mathbfcal{H}_n]
            -i 2 (\vec{v}_z \cdot \mathbfcal{E}^*_n) (\vec{v}_z\cdot[ \nabla \times \mathbfcal{H}_n]),
            \\
            2\omega \: \mathbfcal{B}_n \odot  \mathbfcal{H}^*_n & =
            - k_n \mathbfcal{H}^*_n \cdot [ \mathbfcal{E}_n \times \vec{v}_z ]
            - i [ \nabla \times \mathbfcal{E}_n]  \cdot \mathbfcal{H}^*_n
            +i 2 (\vec{v}_z\cdot[ \nabla \times \mathbfcal{E}_n]) (\vec{v}_z \cdot \mathbfcal{H}^*_n),
            \\
            2\omega \: \mathbfcal{B}^*_n \odot \mathbfcal{H}_n  & =
            - k_n  \mathbfcal{H}_n \cdot [ \mathbfcal{E}^*_n  \times \vec{v}_z]
            + i [ \nabla \times \mathbfcal{E}^*_n] \cdot \mathbfcal{H}_n
            -i 2 ( \vec{v}_z\cdot[ \nabla \times \mathbfcal{E}^*_n]) (\vec{v}_z \cdot \mathbfcal{H}_n ).
        \end{align}
    \end{subequations}
\end{widetext}
Therefore, Eq.~\eqref{seq_orthog_intermed} can be written as
\begin{equation}
   \xi\up{M}_n = \dfrac12 (I_1 + I_2 + I_3 ),
\end{equation}
where we consider each of the contributions $I_n$ in turn.
\begin{enumerate}
    \item  
Taking the first terms from the right-sides of Eqs.~\eqref{seq_detailed}, and applying the identity $\vec{X} \cdot [\vec{Y}\times \vec{Z} ] =  [\vec{X}\times \vec{Y}]\cdot \vec{Z}$, we obtain
\begin{equation}
    I_1 = 2\frac{k_n}{\omega}  \int d\vec{r}_{\perp} ~\Big[ [ \mathbfcal{E}_n \times  \mathbfcal{H}^*_n ]_z  +
    [\mathbfcal{E}^*_n \times  \mathbfcal{H}_n  ]_z   \Big] = 2\frac{k_n}{\omega}   \xi\up{E}_n.
\end{equation}
\item 
The second terms from the right-sides of Eqs.~\eqref{seq_detailed} evaluate to
\begin{equation}
    I_2 = i \int d\vec{r}_{\perp} ~(C_2 - C_2^*) = 2i  \int d\vec{r}_{\perp} ~ \mathrm{Im} [C_2],
\end{equation}
where
$ C_2 =  [ \nabla \times \mathbfcal{H}_n] \cdot \mathbfcal{E}^*_n +  [ \nabla \times \mathbfcal{E}^*_n] \cdot \mathbfcal{H}_n$.
\item 
Similarly, the third terms from Eqs.~\eqref{seq_detailed} are
\begin{equation}
    I_3 = i \int d\vec{r}_{\perp} ~(C_3 - C_3^*) = 2i  \int d\vec{r}_{\perp} ~\mathrm{Im} [C_3],
\end{equation}
where
$ C_3 =  2(\vec{v}_z \cdot \mathbfcal{E}_n) ( \vec{v}_z\cdot[ \nabla \times \mathbfcal{H}^*_n] ) + 2 (\vec{v}_z \cdot [ \nabla \times \mathbfcal{E}_n]) (\vec{v}_z \cdot \mathbfcal{H}^*_n)$.
\end{enumerate}
As a result, the substitution of Eqs.~\eqref{seq_detailed} into Eq.~\eqref{seq_orthog_intermed} gives us
\begin{equation}
    \xi\up{M}_n =  \frac{k_n}{\omega} \xi\up{E}_n + i  \Big(\mathrm{Im} [C_2] + \mathrm{Im} [C_3] \Big).
\end{equation}
Both $\xi\up{M}_n$ and $\xi\up{E}_n$ are real, which proves the equality \eqref{eq_orthogonality_condition_connection},
\begin{equation}
    \frac{\xi\up{M}_n}{\xi\up{E}_n} =  \frac{k_n}{\omega}.
\end{equation}
For a final check, by evaluating $I_2$ and $I_3$ in components, by applying integration by parts to certain terms, and by taking into account that boundary terms vanish due to the localization of the mode profiles, one can verify that $I_2+I_3 = 0$.

\end{document}